\documentclass[11pt]{article}

\usepackage[margin=1.25in]{geometry}
\usepackage{amsmath,appendix}
\usepackage{appendix}
\usepackage{amsthm}
\usepackage{amssymb}
\usepackage{setspace}
\usepackage{color}
\usepackage{float}
\usepackage{dsfont}
\usepackage{enumerate}
\usepackage{hyperref}
\usepackage{graphicx}
\usepackage{sidecap}
\usepackage{subfigure}
\usepackage{algorithm}
\usepackage{algorithmic}
\usepackage{comment}
\newcommand{\commentout}[1]{}
\usepackage[utf8]{inputenc}
\usepackage{tikz}
\usetikzlibrary{plotmarks}
\usepackage{wrapfig}

\newtheorem{proposition}{Proposition}

\newtheorem{theorem}{Theorem}

\newtheorem{lemma}{Lemma}

\newtheorem{assumption}{Assumption}

\numberwithin{equation}{section}
\numberwithin{figure}{section}

\renewcommand{\subset}{\subseteq}
\renewcommand{\hat}{\widehat}

\renewcommand{\epsilon}{\varepsilon}

\def\supp{\text{supp}}
\def\OmegaRF{\text{SRF}}
\def\md{{\rm md}}

\newcommand{\calH}{\mathcal{H}}

\newcommand{\calM}{\mathcal{M}}
\newcommand{\diag}{\text{diag}}

\newcommand{\beq}{\begin{equation}}
\newcommand{\eeq}{\end{equation}}

\newcommand{\vertiii}[1]{{\left\vert\kern-0.25ex\left\vert\kern-0.25ex\left\vert #1 
    \right\vert\kern-0.25ex\right\vert\kern-0.25ex\right\vert}}
    
\def\({\Big(}
\def\){\Big)}
\def\C{\mathbb{C}}

\def\R{\mathbb{R}}
\def\T{\mathbb{T}}
\def\Z{\mathbb{Z}}

\title{Super-resolution limit of the ESPRIT algorithm}

\author{Weilin Li\thanks{Courant Institute of Mathematical Sciences. Email: weilinli@cims.nyu.edu}
\and Wenjing Liao\thanks{Georgia Institute of Technology. Email: wliao60@gatech.edu}
\and Albert Fannjiang\thanks{University of California, Davis. Email: fannjiang@math.ucdavis.edu}}

\begin{document}
	
\maketitle

\begin{abstract}
	The problem of imaging point objects can be formulated as estimation of an unknown atomic measure from its $M+1$ consecutive noisy Fourier coefficients. The standard resolution of this inverse problem is $1/M$ and super-resolution refers to the capability of resolving atoms at a higher resolution. When any two atoms are less than $1/M$ apart, this recovery problem is highly challenging and many existing algorithms either cannot deal with this situation or require restrictive assumptions on the sign of the measure. ESPRIT
	is an efficient method that does not depend on the sign of the measure. This paper provides an explicit error bound on the support matching distance of ESPRIT in terms of the minimum singular value of Vandermonde matrices. 
	When the support consists of multiple well-separated clumps and noise is sufficiently small,  
	 the support error by ESPRIT scales like ${\rm SRF}^{2\lambda+2} \times {\rm Noise}$, where the Super-Resolution Factor (${\rm SRF}$) governs the difficulty of the problem and $\lambda$ is the cardinality of the largest clump. 
	 {If the support contains one clump of closely spaced atoms, the min-max error is ${\rm SRF}^{2\lambda+2} \times {\rm Noise}/M$. Our error bound matches the min-max rate up to a factor of $M$ in the small noise regime. Our results therefore establishes the near-optimality of ESPRIT,} and
	  our theory is validated by numerical experiments.

	\vspace{1em} 
\end{abstract}

\section{Introduction}

\subsection{Background and motivation}

Many imaging problems involve detection of point objects from Fourier measurements. Such inverse problems arise in many interesting applications in imaging and signal processing, including Direction-Of-Arrival (DOA) estimation \cite{krim1996array,schmidt1986multiple}, inverse source and inverse scattering \cite{fannjiang2010remote,fannjiang2015compressive,fannjiang2010compressive}, and time series analysis \cite{stoica1997introduction}. 
The problem can be formulated as spectral estimation - estimating an unknown discrete measure $\mu$ consisting of a collection of Dirac delta functions, from its noisy low-frequency Fourier coefficients.

The first solution to spectral estimation can be traced back to Prony \cite{prony1795essai}.   Unfortunately, the Prony's method is numerically unstable and numerous modifications have been attempted to improve  its  numerical  behavior. In the signal processing community, a class of subspace methods achieved major breakthroughs for the DOA estimation. Important representative subspace methods are MUSIC (MUltiple SIgnal Classification) \cite{schmidt1986multiple}, ESPRIT (Estimation of Signal Parameters via Rotation Invariance Techniques) \cite{kailath1989esprit}, and the matrix pencil method \cite{hua1990matrixpencil}. MUSIC was one of the first robust methods that gives a high-resolution recovery but its computational cost is high. This drawback motivated the development of more efficient algorithms such as ESPRIT and the matrix pencil method. These methods have been widely used in applications due to their high-resolution recovery -- they are capable of resolving fine details in $\mu$ \cite{krim1996array}.

A central interest on the mathematical theory of super-resolution is to understand how to stably estimate $\mu$ when there are closely spaced atoms in $\mu$. Let $\Delta$ be the minimum separation of $\mu$, which is defined as the distance between the two closest atoms in the support of $\mu$. Suppose $M+1$ consecutive noisy Fourier coefficients of $\mu$ are collected. 
The standard resolution of this inverse problem is $1/M$, which is the threshold predicted by the Heisenberg uncertainty principle. Super-resolution estimation refers to the case where $\Delta$ is significantly smaller than $1/M$. In this situation the recovery is very sensitive to noise.

In recent years, the theory of super-resolution has gained considerable attention partly due to the invention of a new family of convex minimization methods for this problem, see \cite{candes2014towards,candes2013super,tang2013offgrid,azais2015spike,duval2015exact,li2017elementary}. While they are successful when $\Delta\geq C/M$ for a reasonably small $C>1$, they can potentially fail when $\Delta\leq 1/M$, even in the noiseless regime. For these convex methods to succeed when $\Delta\leq 1/M$, one requires that $\mu$ is non-negative \cite{morgenshtern2016super,schiebinger2017superresolution}, or more generally, the sign of its atoms satisfies certain algebraic criteria \cite{benedetto2018super}. Hence, it appears that an entirely different approach is required to deal with the case of closely spaced atoms with arbitrary complex phases that are pertinent to many applications.

Subspace methods like MUSIC and ESPRIT are considerably different from the aforementioned convex approaches. First, they do not involve convex optimization. Second, they provide exact recovery when there is no noise, regardless of the location of the atoms, as long as the number of measurements is at least twice the number of atoms. Third, numerical evidence has demonstrated that they can accurately estimate $\mu$ with arbitrarily complex phases, even when $\Delta$ is significantly smaller than $1/M$, provided that the noise level is sufficiently small. In other words, MUSIC and ESPRIT have super-resolution capabilities, regardless of the sign of $\mu$.

An interesting question is to quantify the resolution limit of MUSIC and ESPRIT -- conditions on $\mu$ and the noise level for which they can recover $\mu$ up to a prescribed error.
The answer is not straightforward, as simple numerical experiments show that the stability of MUSIC and ESPRIT heavily depends on how the support of $\mu$ is arranged. In our earlier works \cite{li2017stable,li2019music}, we introduced a separated clumps model to allow for atoms clustered in far apart sets, and proved accurate estimates on the minimum singular values of the Vandermonde matrices with nodes satisfying this separated clumps model. The super-resolution limit of MUSIC was studied in \cite{li2019music}.

This paper focuses on the robustness of ESPRIT. Although this method was invented over a quarter century ago, an accurate analysis of its super-resolution limit has been elusive. 
One main complication is that, one must estimate the minimum singular value of two structured matrices that appear in ESPRIT. One is a rectangular Vandermonde matrix with nodes on the unit circle (which we will denote by $\Phi_M$) and the other one is constructed as part of the algorithm (which we will denote by $U_0$). The minimum singular value of $\Phi_M$ was carefully studied in \cite{li2017stable,batenkov2018conditioning,kunis2018condition}. The conditioning of $U_0$ was not previously analyzed. In the process of studying these matrices, we discover that ESPRIT implicitly leverages an uncertainty principle for non-harmonic Fourier series. This connection has not been previously discovered, and is the key piece that allows us to provide an accurate analysis for the super-resolution limit of ESPRIT. 

\subsection{Contributions and outline}

In Section \ref{secesprit}, we review the ESPRIT algorithm and introduce the necessary notation. ESPRIT has been empirically observed to be robust to noise and is capable of super-resolving atoms with arbitrary spacing and phases, provided that the noise is sufficiently small. This paper rigorously derives the error bound of ESPRIT, and proves the resolution limit of ESPRIT under a geometric model for the unknown support. 

In Section \ref{secrobust}, we derive the main stability bounds for ESPRIT. We first bound the matching distance between the true and estimated supports, in terms of the noise level and the minimum singular value of the Vandermonde matrices. Our bound in Theorem \ref{thm1} significantly improves upon the existing ones in \cite{aubel2016deterministic,aubel2016performance,fannjiang2016compressive}, especially when $\Delta\leq 1/M$. Theorem \ref{thm1} is deterministic, non-asymptotic and holds for any support, including when $\Delta$ is arbitrarily small. Numerical evidence demonstrates that it provides an accurate dependence on the minimum singular value. 

In Section \ref{secsingular}, we combine Theorem \ref{thm1} with bounds on the minimum singular value of Vandermonde matrices under a separated clumps model \cite{li2017stable, li2019music} to obtain Theorem \ref{thmesprit}. This is the first known rigorous guarantee for ESPRIT in the $\Delta\leq 1/M$ regime. The theorem shows that ESPRIT can accurately recover the support of $\mu$, {\it regardless of how small $\Delta$ is}, provided that the noise is smaller than a quantity specified by our theorem. Define the super-resolution factor ${\rm SRF}:=1/(\Delta M)$, which can be interpreted as the maximum number of atoms located within an interval of length $1/M$. We show that if the noise $\epsilon$ is sufficiently small, the support error by ESPRIT is $O( {\rm SRF}^{-(2\lambda-2)}\, \epsilon )$, where $\lambda$ is the cardinality of the largest clump of the measure. When continuous Fourier measurements are collected, and the support contains only one clump of closely spaced atoms, the min-max rate is $O( {\rm SRF}^{-(2\lambda -2)} \, \epsilon/M)$ \cite{batenkov2019super}.
In the small noise regime, our support estimate for ESPRIT matches the min-max rate up to a factor of $M$, which establishes the near-optimality of ESPRIT. 

In Section \ref{secuncertainty}, we derive several new uncertainty principles for discrete non-harmonic Fourier series. A crucial step in our analysis of ESPRIT is to derive a lower bound, uniformly over all $\mu$, on the minimum singular value of $U_0$. We show that the minimum singular value of $U_0$ is related to an uncertain principle in discrete non-harmonic Fourier series. We establish this uncertainty principle in Theorem \ref{thm:UPcomplex} 
which might be of independent mathematical interest. Section \ref{secuncertainty} can be read independently of the rest of the paper. 

Appendix \ref{seclemmas} contains the proof of the results stated in Sections \ref{secrobust} and \ref{secsingular}.

\subsection{Related work}
\label{secrelated}

MUSIC and ESPRIT were originally invented for DOA estimation where the amplitudes of $\mu$ are assumed to be random and multiple snapshots of measurements are taken. In this ``multiple snapshot" setting, more information about the support of $\mu$ is collected, and statistics about the amplitudes of $\mu$ can be utilized. Sensitivity of MUSIC and ESPRIT for the DOA estimation was studied in \cite{swindlehurst1990sensitivity,swindlehurst1992performance,swindlehurst1993performance,li1992sensitivity}. This paper focuses on the ``single snapshot" setting where the amplitudes of $\mu$ are deterministic and little statistical information can be utilized.

Regarding the stability analysis of subspace methods, there have been works on bounding the error in terms of the minimum singular value of Vandermonde matrices. Such inequalities can be found in \cite{liao2016music,li2017stable} for MUSIC and in \cite{fannjiang2016compressive} for ESPRIT, as well as in \cite{moitra2015matrixpencil} for the matrix pencil method. One major roadblock is that, to provide a comprehensive error estimate, one still needs to accurately bound the smallest singular value of $\Phi_M$ in the $\Delta\leq 1/M$. This difficulty was addressed in \cite{li2017stable}, which provided the first accurate analysis of MUSIC in the $\Delta\leq 1/M$ regime.
As for ESPRIT, the bounds in \cite{fannjiang2016compressive,aubel2016deterministic} do not capture the exact dependence of the error on the minimum singular value, and consequently, are inaccurate when $\Delta\leq 1/M$ (see \eqref{eqesprit1} and \eqref{eqesprit2} and the discussion there). 

The minimum singular value of the Vandermonde matrix crucially depends on the configuration of its nodes. One available bound in the case $\Delta\geq C/M$ was proved in \cite{moitra2015matrixpencil}, which relied on the Beurling-Selberg machinery, see \cite{vaaler1985some}. Recently there are several independent works which provide estimates for $\Delta\leq 1/M$ by incorporating additional geometric information about the support set, see \cite{batenkov2018conditioning,li2017stable,kunis2018condition}. Accurate lower bounds under a separated clumps model can be found in \cite{li2017stable}. 

Super-resolution has also been addressed from a statistical or information theoretic point of view. The authors of \cite{donoho1992superresolution,demanet2015recoverability} considered a sparsity prior, where the $S$ atoms are located on a grid on $\mathbb{R}$ with spacing $1/N$ and the measurements consist of noisy continuous Fourier measurements on $[-M,M]$. It was proved in \cite{demanet2015recoverability} that the optimal algorithm can estimate the measure in $\ell^2$ error of $O(  {
\rm SRF}^{2S-1} \epsilon)$ where ${
\rm SRF} = N/M$ and $\epsilon$ is the noise level. 

More relevant to our work is \cite{batenkov2019super}, where the authors considered a geometric prior with one clump of closely spaced atoms. Let $\lambda$ be the cardinality of the largest clump or cluster. They proved that the optimal algorithm can estimate the support with accuracy of $O( {\rm SRF}^{2\lambda-2} \epsilon/M)$. These results are information theoretic, and they do not provide a tractable algorithm that achieves these rates\footnote{Upon closer inspection of their proofs, it is implicit that $\ell^0$ minimization is optimal, but this is not a computationally feasible method.}. In comparison, ESPRIT is a polynomial-time algorithm and the results in this paper show that ESPRIT is near optimal for the geometric clumps model, since it achieves the support recovery rate of $O( {\rm SRF}^{2\lambda -2} \epsilon)$.


\section{Review of ESPRIT}
\label{secesprit}

We first describe the spectral estimation problem. Let $\calM_S$ be the collection of non-zero and complex-valued discrete measures on the periodic unit interval $\mathbb{T} = [0,1)$ with at most $S$ atoms, and let $\delta_{\omega}$ denote the Dirac measure supported in $\omega$. Any $\mu\in\calM_S$ is of the form,
\begin{equation*}
\label{eq:model1}
\mu(\omega)=\sum_{j=1}^S x_j\delta_{\omega_j}(\omega)
\quad\text{where}\quad  
x:=\{x_j\}_{j=1}^S \in \mathbb{C}^S
\quad\text{and}\quad
\Omega:=\{\omega_j\}_{j=1}^S\subset\T. 
\end{equation*}
The minimum separation of $\mu$ is defined as
\begin{equation*}
\label{eq:minsep}
\Delta
:=\min_{j\not=k} |\omega_j-\omega_k|_\T
:=\min_{j\not=k} \min_{n\in\Z} |\omega_j-\omega_k-n|.
\end{equation*}
Let $y^0=\{y_k^0\}_{k=0}^M\in\C^{M+1}$ denote the first $M+1$ consecutive Fourier coefficients of $\mu$:
\begin{equation*}
\label{eq:model2}
y_k^0
:=\hat\mu(k)
:=\int_{\T} e^{-2\pi ik\omega} \ d\mu(\omega)
=\sum_{j=1}^S x_j e^{-2\pi i  k\omega_j}, \quad \text{for } k=0,1,\ldots,M. 
\end{equation*}
Suppose we are given information about $\mu\in\calM_S$ in the form of $M+1$ consecutive noisy Fourier coefficients, 
\begin{equation*}
\label{eq:model3}
y:=y^0+\eta, 
\end{equation*}
where $\eta\in\C^{M+1}$ represents some unknown noise vector. We let $\|\cdot\|_2$ denote either the Euclidean norm or the spectral norm, and we use $\|\cdot\|_F$ for the Frobenius norm. 

The $(M+1)\times S$ Fourier or Vandermonde matrix whose nodes are specified by $\Omega$ is denoted
\begin{equation}
\label{eq:phi}
\Phi_M:=\Phi_M(\Omega)
:=
\begin{bmatrix}
1 &1 &\cdots &1 \\
e^{-2\pi i\omega_1} &e^{-2\pi i\omega_2} &\cdots &e^{-2\pi i\omega_S} \\
\vdots &\vdots & &\vdots \\
e^{-2\pi i M\omega_1} &e^{-2\pi iM\omega_2} &\cdots &e^{-2\pi iM\omega_S}
\end{bmatrix}.
\end{equation}
If $\mu\in\calM_S$ has amplitudes $x$ and support $\Omega$, then we have the relationship
\begin{equation}
\label{eq:model4}
y = \Phi_Mx+\eta. 
\end{equation} 

The goal of spectral estimation is to stably recover $\mu$, including the support $\Omega$ and the amplitudes $x$, from $y$. A typical two-step strategy is to estimate the support set and then the amplitudes. 
ESPRIT exploits the Vandermonde decomposition of a Hankel matrix in order to reformulate the support estimation step as an eigenvalue problem.  Throughout the exposition, $L$ is an integer parameter for ESPRIT that satisfies
\begin{equation}
\label{eq:L}
S\leq L \leq M+1-S.
\end{equation}
Note that it always possible to find a $L$ that satisfies the above inequalities whenever the number of measurements exceeds the amount of unknowns: $M+1 \ge 2S$. The Hankel matrix of $y$ (with parameter $L$) is defined to be
\begin{equation}
\label{eq:hankely}
\calH(y) 
:= 
\begin{bmatrix}
y_0 & y_1 & \hdots & y_{M-L}
\\
y_1 & y_2 & \hdots & y_{M-L+1}
\\
\vdots & \vdots & \ddots & \vdots
\\
y_{L} & y_{L+1} & \hdots & y_{M}
\end{bmatrix} 
\in \C^{(L+1) \times (M-L+1)}.
\end{equation}

We first describe ESPRIT in the noiseless setting and then outline how it deals with noise. In the case where $\eta=0$, we have access to the Hankel matrix $\calH(y^0)$, and a direct calculation shows that $\calH(y^0)$ processes the following Vandermonde decomposition:
\begin{equation}
\label{eq:Vfact}
\calH(y^0)  
= \Phi_L D_X \Phi_{M-L}^T, 
\end{equation}
where $D_X = \diag(x_1,\dots,x_S) \in \mathbb{C}^{S \times S}$. The conditions in \eqref{eq:L} imply that both $\Phi_L$ and $\Phi_{M-L}$ have full column rank, which in turn implies that $\calH(y^0)$ has rank $S$. More importantly, we have ${\rm Range}(\calH(y^0)) = {\rm Range}(\Phi_L)$, which means ${\rm Range}(\calH(y^0))$ contains full information about the column span of $\Phi_L$. ESPRIT amounts to finding an orthonormal basis of ${\rm Range}(\calH(y^0))$ and using this basis to recover $\Omega$. The procedure of finding an orthonormal basis of ${\rm Range}(\calH(y^0))$ can be realized by Singular Value Decomposition (SVD) or QR decomposition of $\calH(y^0)$.

Let the SVD of $\calH(y^0)$ be
\begin{equation}
\label{eq:Hsvd}
\calH(y^0) = [\underbrace{ U}_{(L+1) \times S} \  \underbrace{ {\ U_\perp} }_{(L+1) \times (L+1-S)}] \underbrace{\Sigma}_{(L+1) \times (M-L+1)}  [\underbrace{ V}_{(M-L+1) \times S} \  \underbrace{ {\ V_\perp} }_{(M-L+1) \times (M-L+1-S)}]^*,
\end{equation} 
where $\Sigma={\rm diag}( \sigma_1 , \ldots ,  \sigma_S, 0 ,\ldots,0)$ contains the singular values of $\calH(y^0)$. In general, we let $\sigma_j(\cdot)$ denote the $j$-th largest singular value of a matrix. Comparing the identities \eqref{eq:Vfact} and \eqref{eq:Hsvd}, we see that the column space of $U$ and $\Phi_L$ are identical to ${\rm Range}(\calH(y^0))$. There exists an invertible matrix $P \in \mathbb{C}^{S \times S}$ 
such that 
\begin{equation}
\label{eq:P}
U = \Phi_L P\in\C^{(L+1)\times S}.
\end{equation}
Let $U_0$ and $U_1$ be two submatrices of $U$ containing the first and the last $L$ rows respectively. Then we have 
\begin{align*}
U_0 &= \Phi_{L-1} P, \\
U_1 &= \Phi_{L-1} D_\Omega P , 
\end{align*}
where $D_\Omega = \diag(e^{-2\pi i \omega_1},\dots e^{-2\pi i\omega_S})$. Setting $L \ge S$ as \eqref{eq:L} guarantees that $U_0$ and $U_1$ have full column rank.

It follows from these definitions that if we define the matrix $\Psi := U_0^{\dagger} U_1$, then 
\begin{equation}
\label{eq:P2}
\Psi = P^{-1} D_\Omega P\in\C^{S\times S}. 
\end{equation}
Hence, the eigenvalues of $\Psi$ are exactly $\{e^{-2\pi i \omega_j}\}_{j=1}^S$. The ESPRIT technique amounts to finding the support set $\Omega$ through the eigenvalues of $\Psi$. 

In the presence of noise, the ESPRIT algorithm forms the noisy Hankel matrix 
\[
\calH(y)
=\calH(y^0)+\calH(\eta). 
\]
If the noise is sufficiently small, then the rank of $\calH(y)$ is at least $S$. ESPRIT computes a matrix $\hat\calH(y)$, defined to be the best rank $S$ approximation of $\calH(y)$ in the spectral norm; this amounts to computing the SVD of $\calH(y)$ and truncating the singular spaces. We write the SVD of $\calH(y)$ in \eqref{svdnoisy}. 

When the size of the noise is sufficiently small, we expect the column space of $\hat U$ to be a small perturbation of that of $U$. The ESPRIT algorithm proposes to find the eigenvalues of 
\[
\hat\Psi
=\hat U_0^\dagger \hat U_1,
\]
where $\hat U_0$ and $\hat U_1$ are the first and last rows of $\hat U$ respectively. Projecting the eigenvalues to the complex unit circle provide us with an estimator $\hat\Omega$ for $\Omega$. Further details can be found in Algorithm \ref{alg:ESPRIT}.

\renewcommand{\algorithmicrequire}{\textbf{Input:}}
\renewcommand{\algorithmicensure}{\textbf{Output:}}
\begin{algorithm}[t!]                      	
	\caption{ESPRIT}          	
	\label{alg:ESPRIT}		
	\begin{algorithmic}[1]                    	
		\REQUIRE  $y \in \C^{M+1}$, sparsity $S$, $L$
		\begin{enumerate}
			\item 
			 Form Hankel matrix $\calH(y) \in \C^{(L+1)\times (M-L+1)}$
			\item Compute the SVD of $\calH(y)$: 
			\begin{equation}
			\label{svdnoisy}
			\hspace{-4em}
			\calH(y) = [\underbrace{\hat U}_{(L+1) \times S} \  \underbrace{{\hat U}_\perp}_{(L+1) \times (L+1-S)}] \underbrace{\hat\Sigma}_{(L+1) \times (M-L+1)}[\underbrace{\hat V}_{(M-L+1) \times S} \  \underbrace{{\hat V}_\perp}_{(M-L+1) \times (M-L+1-S)}]^*  
			\end{equation}
			where $\hat\Sigma=\diag(\hat\sigma_1,\dots,\hat \sigma_S, \hat \sigma_{S+1}, \ldots)$ are the singular values of $\calH(y)$ listed in decreasing order.
		\end{enumerate}
		\begin{enumerate}
			\item[3.] Let $\hat U_0$ and $\hat U_1$ be two submatrices of $\hat U$ containing the first and the last $L$ rows respectively. Compute 
			$$\hat\Psi = \hat U_0^{\dagger} \hat U_1$$
			and its $S$ eigenvalues $\hat \lambda_1,\ldots,\hat \lambda_S$.                
			\item[Output:] $\hat \Omega = \{\hat\omega_j\}_{j=1}^S$ where $\hat\omega_j = -\frac{\angle \hat \lambda_j}{2\pi }$.
		\end{enumerate} 
	\end{algorithmic}
\end{algorithm}

\section{Robustness of ESPRIT}
\label{secrobust}

A central interest about ESPRIT is on it stability analysis. The main goal of this section is to bound the error between $\Omega$ and $\hat\Omega$ in terms of the matrices that appear in the ESPRIT algorithm. We state the main theorems and the necessary lemmas in this section and leave their proofs in Appendix \ref{seclemmas}.

\subsection{Perturbation of the matrix $\Psi$ in spectral norm} 

Before we proceed to the stability analysis, we need to point out a subtle and important feature of ESPRIT. The singular values and singular subspaces of a matrix are unique, but the SVD only provides us with one of infinitely many equivalent orthonormal bases. Importantly, ESPRIT is invariant to the specific choice of orthonormal basis for the column span of $\widehat U$. In other words, the eigenvalues of $\hat\Psi$ remain the same if one uses another orthonormal basis for the column span of $\widehat U$. 
To see why, let $\widetilde U$ be another orthonormal basis for the column span of $\hat U$. Then there exists an invertible matrix $R \in \C^{S \times S}$, such that $\widetilde U = \hat U R$. Let $\widetilde U_0$ and $\widetilde U_1$ be two submatrices of $\widetilde U$ containing the first and the last $L$ rows respectively. Then $\widetilde U_0 = \hat U_0 R$ and $\widetilde U_1 = \hat U_1 R$. It follows that
$\widetilde U_0^\dagger \widetilde U_1 = R^{-1} \hat U_0^\dagger \hat U_1 R$, so the eigenvalues of $\widetilde U_0^\dagger \widetilde U_1$ are identical to those of $\hat U_0^\dagger \hat U_1$.

It follows from the above observation that we can make the following reduction. The output of ESPRIT is independent of the particular choice of basis for the singular spaces, so for the mathematical analysis, we can without loss of generality, select particular matrices $U$ and $\hat U$ that are most suitable for our analysis. It turns out that the most convenient choice is when the columns of $U$ and $\hat U$ consist of the canonical vectors \footnote{\url{https://en.wikipedia.org/wiki/Angles_between_flats}}. 
Our first perturbation bound is on $\|\hat\Psi -\Psi\|_2$ (see Appendix \ref{secapppsi} for the proof).

\begin{lemma}
\label{lemmapsi}
Fix positive integers $L,M,S$ such that \eqref{eq:L} holds. For any $\mu\in\calM_S$ and $\eta\in\C^{M+1}$, if 
$$\|\calH(\eta)\|_2 \le 
\frac{x_{\min}\sigma_{S}(\Phi_L)\sigma_{S}(\Phi_{M-L})\sigma_S(U_0)}{4 \sqrt{2S}},$$
 then
	\[
	\|\widehat\Psi - \Psi\|_2 
	\le 
	\frac{14\sqrt{2S}\|\calH(\eta)\|_2}{x_{\min}\sigma_{S}(\Phi_L)\sigma_{S}(\Phi_{M-L})\sigma_S^2(U_0)}
	.
	\]	
\end{lemma}

The estimate in Lemma \ref{lemmapsi} above indicates that the smallest singular values of $\Phi_L$ and $U_0$ play an important role in our analysis. An accurate lower bound of $\sigma_S(\Phi_L)$ has been proved in \cite{li2017stable}  under a separated clumps model of nodes. We will derive a lower bound of $\sigma_S(U_0)$ in Subsection \ref{secu0}.

\subsection{ESPRIT error in the support matching distance}

We next relate $\|\hat \Psi -\Psi\|$ to the matching distance between the eigenvalues of $\hat\Psi$ and $\Psi$. The matrix $\Psi$ is diagonalizable and has eigenvalues $\{e^{-2\pi i \omega_j}\}_{j=1}^S$. We let $\{\hat\lambda_j\}_{j=1}^S$ be the eigenvalues of $\hat\Psi$ listed according to multiplicity. It will follow from the results below that the eigenvalues of $\hat\Psi$ are distinct provided that the noise is sufficiently small. The \textit{matching distance} between the eigenvalues of $\Psi$ and $\hat\Psi$ is defined to be
$${\rm md}(\Psi,\hat\Psi) := \min_{\psi} \max_{j} |\hat\lambda_{\psi(j)}-e^{-2\pi i \omega_j}|,$$
where $\psi$ is taken over all permutations of $\{1,\ldots,S\}$.

ESPRIT projects each eigenvalue $\hat\lambda_j$ of $\hat\Psi$ to the complex unit circle to obtain $\hat\omega_j = -{\angle \hat \lambda_j}/(2\pi )\in\T$. Let $\hat\Omega=\{\hat\omega_j\}_{j=1}^S$, which is the output of ESPRIT. Notice that $\hat\Omega$ may contain repeated entries since $\hat\Psi$ may have eigenvalues with multiplicity greater than one. If the noise is sufficiently small, then $\hat\Omega$ necessarily consists of $S$ distinct values. The matching distance between $\Omega$ and $\hat\Omega$ is,
\[
{\rm md}(\Omega,\hat\Omega) 
:= \min_{\psi} \max_{j} |\hat\omega_{\psi(j)}-\omega_j|_{\mathbb{T}}.
\]
The two matching distances ${\rm md}(\Omega,\hat\Omega)$ and ${\rm md}(\Psi,\hat\Psi)$ satisfy the following relation (proved in Appendix \ref{appeq}):
\begin{equation}
\label{lemmamatching3}
	{\rm md}(\Omega,\hat\Omega)
	\leq \frac{1}{2} {\rm md}(\Psi,\hat\Psi).
\end{equation}

We provide two bounds on the matching distance. The first bound in Lemma \ref{lemmamatching} (a) below is based on the Bauer-Fike theorem. It holds for various noise levels, but leads to a weaker conclusion. The second bound in Lemma \ref{lemmamatching} (b) below is derived from the Gershgorin circle theorem. In comparison to the first bound, the second one requires the noise level to be smaller but the conclusion is stronger. Their proofs are in Appendix \ref{appmd}.

\begin{lemma}
\label{lemmamatching}
Fix positive integers $L,M,S$ such that \eqref{eq:L} holds. For any $\mu\in\calM_S$ and $\eta\in\C^{M+1}$, the following hold.
	\begin{enumerate}[(a)]
	
	\item 
	(Moderate noise regime) We have
	\begin{equation}
	\label{eqmatching1}
	{\rm md}(\Omega,\hat\Omega) \leq \frac{S^{3/2}\sqrt{L+1}}{\sigma_{S}(\Phi_L)}\|\hat\Psi-\Psi\|_2. 
	\end{equation}
			
	\item 
	(Small noise regime) If additionally,
	\begin{equation}
	\label{eq:smallnoise}
	\|\hat\Psi-\Psi\|_2
	\leq \frac{\sigma_{S}^2(\Phi_L) \Delta}{S^2 (L+1)},
	\end{equation}
	then we obtain
	\begin{equation}
	{\rm md}(\Omega,\hat\Omega)
	\leq 
	\|\Psi-\hat\Psi\|_2. 
	\label{eqmatching2}
	\end{equation}		

	\end{enumerate}	

\end{lemma}

Part (a) holds without any assumptions on $\|\hat\Psi-\Psi\|_2$, while Part (b) requires $\|\hat\Psi-\Psi\|_2$ which is in turn a constraint on the size of the noise. However, the conclusion given by \eqref{eqmatching2} is stronger than \eqref{eqmatching1} in the super-resolution regime, because when $\Delta$ is small, $\sigma_{S}(\Phi_L) $ is an extremely small quantity that depends exponentially on ${\rm SRF}^{-1}$.

\subsection{Smallest singular value of $U_0$}
\label{secu0}

Lemma \ref{lemmapsi} indicates that $\sigma_S(U_0)$ plays an important role in our stability analysis for ESPRIT. The matrix $U_0$ is obtained from $U$ with the last row removed, where $U \in \mathbb{C}^{(L+1) \times S}$ contains orthonormal columns. One can easily show that $\sigma_1(U_0) = \ldots = \sigma_{S-1}(U_0)$, but it is not clear what the value of $\sigma_S(U_0)$ is.
In general, deleting a row from a matrix with orthonormal columns may result in linearly dependent columns. For instance, the matrix 
$ \left[ 0 \ 0 ; 1 \ 0 ; 0 \ 1 \right] $
has orthonormal columns, but if we delete its last row, the resulting matrix does not have full column rank any more. 

However, $U$ is not an arbitrary unitary matrix because the columns of $U$ form an orthonormal basis for the column space of $\Phi_L$. While $U$ can be explicitly realized by the Gram-Schmidt orthogonalization process applied to $\Phi_L$, it is hard to leverage this relationship in a theoretical form. 
Instead we establish an uncertainty principle to relate $U$ and $\Phi_L$ in order to prove the following lower bound of $\sigma_S(U_0)$. 

\begin{lemma}
	\label{lemmauncertainty}
	Fix positive integers $L,M,S$ such that \eqref{eq:L} holds. For any $\mu\in\calM_S$, we have
	\[
	\min\big(\sigma_{S}^2(U_0),\sigma_{S}^2(U_1)\big)
	\geq \max\(1-\frac{S}{\sigma_S^2(\Phi_L)},\ 4^{-S} \). 
	\]	
\end{lemma}

Lemma \ref{lemmauncertainty} is proved in Appendix \ref{appu0}, which shows that controlling $\sigma_{S}(U_0)$ is equivalent to establishing an uncertainty principle. Notice that Lemma \ref{lemmauncertainty} gives a deterministic (hence worst case) bound. In principle it is possible to improve it by incorporating further assumptions about $\Omega$.
	
\subsection{Stability of ESPRIT in terms of Vandermonde matrices}

Combining Lemma \ref{lemmapsi}, Lemma \ref{lemmamatching} and Lemma \ref{lemmauncertainty} gives rise to the following deterministic bounds for the support error of ESPRIT (proved in Appendix \ref{seclemmas}), one for moderate noise and one for small noise. 

\begin{theorem}
	\label{thm1}
	Fix positive integers $L,M,S$ such that \eqref{eq:L} holds, and fix $\mu\in\calM_S$ and $\eta\in\C^{M+1}$. 
	\begin{enumerate}[(a)]
		\item 
		(Moderate noise regime) If the noise level is moderately small such that 
		\begin{equation}
		\|\calH(\eta)\|_2  \le \frac{x_{\min}\sigma_S(U_0) \sigma_{S}(\Phi_L) \sigma_{S}(\Phi_{M-L})}{4\sqrt {2 S} }, 
		\label{thm1con}
		\end{equation}
		then the output of ESPRIT satisfies 
		\begin{equation}
		\label{thm1md1}
		\md(\Omega,\hat \Omega)
		\le \frac{20\, S^2 \sqrt{L+1}\,\|\calH(\eta)\|_2}{x_{\min}\sigma_S^2(U_0) \sigma_S^2(\Phi_L) \sigma_S(\Phi_{M-L})}.
		\end{equation}
		\item 
		(Small noise regime) If the noise level is sufficiently small such that
		\begin{equation}
		\label{thm1con2}
		\|\calH(\eta)\|_2  \le \frac{ x_{\min} \Delta \sigma_S^2(U_0) \sigma_{S}^3(\Phi_L) \sigma_{S}(\Phi_{M-L})}{20 S^{5/2} (L+1)},
		\end{equation}
		then the output of ESPRIT satisfies  
		\begin{align}
		\label{thm1md2}
		\hspace{-2em}
		{\rm md}(\Omega,\hat \Omega)
		&\le \frac{20\sqrt{S} \, \|\calH(\eta)\|_2}{x_{\min}\sigma_S^2(U_0) \sigma_S(\Phi_L) \sigma_S(\Phi_{M-L}) }.
		\end{align}
			\end{enumerate}
\end{theorem}

In practice one should choose $L = \lfloor  M/2 \rfloor $ to balance $\sigma_S(\Phi_L)$ and $\sigma_S(\Phi_{M-L})$. If $M$ is even, we set $L = M/2$ so that $\sigma_S(\Phi_L) = \sigma_S(\Phi_{M-L}) = \sigma_{S}(\Phi_{M/2})$.  
Thanks to Lemma \ref{lemmauncertainty}, $\sigma_{S}(U_0) \ge 2^{-S}$ which is independent of $\Omega$. The key to understand the super-resolution of ESPRIT is to obtain a sharp dependence on $\sigma_{S}(\Phi_{M/2})$. Suppressing all terms that depend on $S$ and the amplitudes of $\mu$, but independent of $M$ and $\Omega$, the above Theorem \ref{thm1} is summarized in the following table. 

\begin{center}

	\begin{tabular}{|l|l|}
		\hline 
		\text{Noise assumption} &\text{ESPRIT Error} \\ \hline
		$\|\calH(\eta)\|_2 \lesssim \sigma_S^2(\Phi_{M/2})$ &$
		{\rm md}(\Omega,\hat \Omega)
		\lesssim  \sqrt{M} \, \sigma_S^{-3}(\Phi_{M/2}) \|\calH(\eta)\|_2 $ \\ \hline
		$\|\calH(\eta)\|_2 \lesssim  \,  \Delta \sigma_S^4(\Phi_{M/2})/M$ &${\rm md}(\Omega,\hat \Omega) 
		\lesssim  \sigma_S^{-2}(\Phi_{M/2}) \|\calH(\eta)\|_2$ \\ \hline
	\end{tabular}
	
\end{center}

In terms of the dependence on $ \sigma_S(\Phi_{M/2})$, Theorem \ref{thm1} greatly improves upon earlier stability bounds\footnote{The paper \cite{fannjiang2016compressive} analyzed a variation of the classical ESPRIT algorithm.} of ESPRIT:
\begin{align}
{\rm md}(\Omega,\hat \Omega)
&\lesssim \frac{ \|\eta\|_2} { \sigma_S^5(\Phi_{M/2})}\  \quad \cite[\text{Theorem 1}]{aubel2016deterministic}, \text{proof in } \cite{aubel2016performance}
\label{eqesprit1}\\
\|\hat\Psi -\Psi\|_2 &\lesssim \frac{ \|\calH(\eta)\|_2} { \sigma_S^4(\Phi_{M/2})}\  \quad \cite[\text{Theorem 4}]{fannjiang2016compressive}.
\label{eqesprit2}
\end{align}
We will later see that the exponent of $\sigma_S(\Phi_{M/2})$ is crucial since it determines the exponent of ${\rm SRF}$ in the super-resolution analysis for ESPRIT.

\section{Super-resolution limit of ESPRIT}
\label{secsingular}

This section is devoted to the analysis of the resolution limit of ESPRIT under a separated clumps assumption on $\Omega$, where this model was proposed in \cite{li2017stable,li2019music}.

\subsection{The minimum singular value of $\Phi_M$}
We first define the separated clumps model (see \cite{li2017stable,li2019music} for more details).

\begin{assumption}[Separated clumps model]
	\label{def:clumps}
	Let $M$ and $A$ be a positive integers and $\Omega\subset\T$ have cardinality $S$. We say that $\Omega$ consists of $A$ {\it separated clumps} with parameters $(M,S,\alpha,\beta)$ if the following hold.
	\begin{enumerate}
		\item 
		$\Omega$ can be written as the union of $A$ disjoint sets $\{\Lambda_a\}_{a=1}^A$, where each {\it clump} $\Lambda_a$ is contained in an interval of length $1/M$. 
		\item 
		$\Delta\geq\alpha/M$ with $\max_{1\leq a\leq A} (\lambda_a-1) < {1}/{\alpha}$ where $\lambda_a$ is the cardinality of $\Lambda_a$.
		\item 
		If $A>1$, then the distance between any two clumps is at least $\beta/M$.
	\end{enumerate}
\end{assumption}

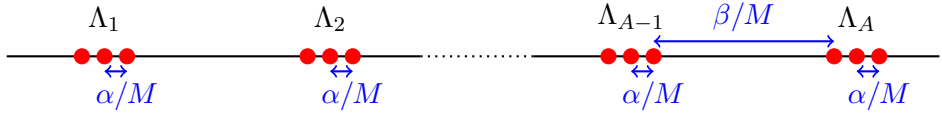
\begin{figure}[h]
	\centering
	\begin{tikzpicture}[xscale = 1,yscale = 1]
	\draw[thick] (-6,0) -- (-0.5,0);
	\filldraw[red] (-5,0) circle (0.1cm);		
	\filldraw[red] (-4.7,0) circle (0.1cm);		
	\filldraw[red] (-4.4,0) circle (0.1cm);		
	\draw[blue,thick,<->] (-4.7,-0.2) -- (-4.4,-0.2);
	\node[blue,below] at (-4.4,-0.2) {$\alpha/M$};
	
	
	\node[above] at (-4.7,0.2) {$\Lambda_1$};
	\filldraw[red] (-2,0) circle (0.1cm);		
	\filldraw[red] (-1.7,0) circle (0.1cm);		
	\filldraw[red] (-1.4,0) circle (0.1cm);		
	\draw[blue,thick,<->] (-1.7,-0.2) -- (-1.4,-0.2);
	\node[blue,below] at (-1.4,-0.2) {$\alpha/M$};
	\node[above] at (-1.7,0.2) {$\Lambda_2$};
	\draw[dotted,thick] (-0.5,0) -- (1,0);
	\draw[thick] (1,0) -- (6.4,0);
	\filldraw[red] (2,0) circle (0.1cm);		
	\filldraw[red] (2.3,0) circle (0.1cm);		
	\filldraw[red] (2.6,0) circle (0.1cm);		
	\draw[blue,thick,<->] (2.3,-0.2) -- (2.6,-0.2);
	\node[blue,below] at (2.6,-0.2) {$\alpha/M$};
	\node[above] at (2.3,0.2) {$\Lambda_{A-1}$};
	\filldraw[red] (5,0) circle (0.1cm);		
	\filldraw[red] (5.3,0) circle (0.1cm);		
	\filldraw[red] (5.6,0) circle (0.1cm);		
	\draw[blue,thick,<->] (5.3,-0.2) -- (5.6,-0.2);
	\node[blue,below] at (5.6,-0.2) {$\alpha/M$};
	\node[above] at (5.3,0.2) {$\Lambda_A$};
	\draw[blue,thick,<->] (2.6,0.2) -- (5,0.2);
	\node[blue,above] at (3.8,0.2) {$\beta/M$};
	\end{tikzpicture}
	\caption{$\Omega = \cup_a \Lambda_a$ where each $\Lambda_a$ contains 3 equally spaced atoms with spacing $\alpha/M$. The clumps are separated at least by $\beta/M$. 
	}
	\label{FigDemoClumps1}
\end{figure}

An example of separated clumps is shown in Figure \ref{FigDemoClumps1}. In applications there are many types of discrete sets that consist of separated clumps. One extreme example is when $\Omega$ is a single clump containing all $S$ points. This is considered to be the worst case configuration for $\Omega$ in the sense that super-resolution will be highly sensitive to noise. Another extreme instance is when all $S$ points in $\Omega$ are separated by $1/M$, so we can think of $\Omega$ as having $S$ clumps each containing single point. This is widely considered to be the best case scenario, in which super-resolution is least sensitive to noise. While our assumption applies to both extremes, the in-between case where $\Omega$ consists of several clumps each of modest size is the most interesting, and developing a theory of super-resolution for this case is most challenging.

We proved in \cite{li2017stable,li2019music} that, under this separated clumps model, $\sigma_{\min}(\Phi_M)$ is an $\ell^2$ aggregate of $A$ terms, where each term only depends on the ``geometry" of each clump. 

\begin{theorem}
	\label{thmsin}
	Fix positive integers $M$ and $S$ such that $M\geq S^2$. Assume $\Omega$ satisfies Assumption \ref{def:clumps} with parameters $(M,S,\alpha,\beta)$ for some $\alpha>0$ and 
	\begin{equation}
	\label{eq:sep2}
	\beta\geq \max_{1\leq a\leq A} \frac{20S^{1/2}\lambda_a^{5/2}}{\alpha^{1/2}}.
	\end{equation}
	Then there exist explicit constants $C_a := C_a(\lambda_a,M)>0$ such that
	\begin{equation}
	\hspace{-.5em} \sigma_{\min}(\Phi_M)
	\geq \sqrt{M}\(\sum_{a=1}^A \big( C_a \alpha^{-\lambda_a+1} \big)^2 \)^{-\frac{1}{2}}. 
	\end{equation}
\end{theorem}

The main feature of this theorem are the exponents on $\OmegaRF=1/\alpha$, which depend on the cardinality of each clumps as opposed to the total number of points. Let $\lambda$ be the cardinality of the largest clump: $\lambda = \max_{1\leq a\leq A} \lambda_a$.

Theorem \ref{thmsin} implies the following bound (which is looser, but easier to digest)
\begin{equation}
\label{eqlowert1}
\sigma_{\min}(\Phi_M) \ge
C \sqrt{M}\ \OmegaRF^{-\lambda+1}.
\end{equation}
Previous results \cite{donoho1992superresolution,demanet2015recoverability} strongly suggest\footnote{We avoid using the word ``imply" because those papers studied a similar inverse problem but with continuous Fourier measurements and bounded measures on $\R$, rather than discrete ones and bounded measures on $\T$, like the ones considered here.} that
\begin{equation}
\label{eqlowert2}
\sigma_{\min}(\Phi_M)\geq C\sqrt{M}\ \OmegaRF^{-S+1}. 
\end{equation}
By comparing the inequalities \eqref{eqlowert1} and \eqref{eqlowert2}, we see the former is dramatically better when all of the point sources are not located within a single clump. These results are also consistent with our intuition that $\sigma_{\min}(\Phi_M)$ is smallest when $\Omega$ consists of $S$ closely spaced points; more details about this can be found in \cite{li2017stable}.

\subsection{Super-resolution limit of ESPRIT}
\label{secSResprit}

Theorem \ref{thm1} provides an error bound of ESPRIT in terms of the singular values of Vandermonde matrices. 
Thanks to Theorem \ref{thmsin}, we can obtain the following explicit error bound for ESPRIT under the separated clumps model (see Appendix \ref{secproofthmesprit} for the proof).

\begin{theorem}
	\label{thmesprit}
	Fix positive integers $M$ and $S$ such that $M\geq S^2$. Suppose $M$ is even and $L=M/2$, and $\Omega$ satisfies Assumption \ref{def:clumps} with parameters $(M/2,S,\alpha,\beta)$  for some $\alpha>0$ and $\beta$ satisfying \eqref{eq:sep2}. Then there exist explicit constants $c_a:=C_a(\lambda_a,M/2)$ such that the following hold:
	\begin{enumerate}[(a)]
		\item 
		(Moderate noise regime) if the noise level is moderately small such that 
		\begin{equation}
		\|\calH(\eta)\|_2  \le \frac{ x_{\min} M  
		}{16 \sqrt { S}   2^{S}} \(\sum_{a=1}^A \big( c_a \alpha^{-\lambda_a+1} \big)^2 \)^{-1}, 
		\end{equation}
		then the output of ESPRIT satisfies 
		\begin{equation}
		\md(\Omega,\hat \Omega)
		\le \frac{80\, S^2 4^S }{x_{\min}M }  \(\sum_{a=1}^A \big( c_a \alpha^{-\lambda_a+1} \big)^2 \)^{\frac 3 2}  \|\calH(\eta)\|_2;
		\end{equation}
		\item 
		(Small noise regime) If the noise level is sufficiently small such that
		\begin{equation}
		\|\calH(\eta)\|_2  \le \frac{\alpha x_{\min}  }{80 S^{5/2} 4^S} \(\sum_{a=1}^A \big( c_a \alpha^{-\lambda_a+1} \big)^2 \)^{-2},
		\end{equation}
		then the output $\hat\Omega$ of ESPRIT satisfies  
		\begin{align}
		\hspace{-2em}
		{\rm md}(\Omega,\hat \Omega)
		&\le \frac{40\sqrt{S} 4^S}{x_{\min}  M}
		\(\sum_{a=1}^A \big( c_a \alpha^{-\lambda_a+1} \big)^2 \)
		 \|\calH(\eta)\|_2.
		\end{align}
			\end{enumerate}
\end{theorem}

To simplify the results, let $\lambda$ be the cardinality of the largest clump. According to \cite[Theorem 2]{li2017stable}, the constant $c_a$ can be expressed as $c_a = C(\lambda_a) (\frac{M}{2\lambda_a})^{\lambda_a-1} \lfloor \frac{M}{2\lambda_a}\rfloor^{-(\lambda_a-1)} $  which weakly depends on $M$ and is bounded above independently of $M$ when $M\geq 4S$.  Notice that ${\rm SRF} = 1/\alpha$. Theorem \ref{thmesprit} can then be simplified to
\begin{center}
	\begin{tabular}{|l|l|}
		\hline 
		\text{Noise assumption} &\text{ESPRIT Error} \\ \hline
		$\|\calH(\eta)\|_2 \lesssim M \, {\rm SRF}^{-(2\lambda-2)}$ &$
		{\rm md}(\Omega,\hat \Omega)
		\lesssim  {\rm SRF}^{3\lambda-3} \|\calH(\eta)\|_2 /M $ \\ \hline
		$\|\calH(\eta)\|_2 \lesssim  \,  {\rm SRF}^{-(4\lambda-3)}$ &${\rm md}(\Omega,\hat \Omega) 
		\lesssim {\rm SRF}^{2\lambda-2}  \|\calH(\eta)\|_2/ M$ \\ \hline
	\end{tabular}
\end{center}
Here, the implicit constants are independent of $M$ and $\Omega$, but depend on $S$, $\{\lambda_a\}_{a=1}^A$ and the amplitudes $x$.

\subsection{In relation to min-max bounds}
\label{secminmax}

In order to compare our results with the min-max error \cite{batenkov2019super}, we next express the noise level in terms of $\|\eta\|_\infty$ using the inequality $\|\calH(\eta)\|_2\leq \|\calH(\eta)\|_F\leq M \|\eta\|_\infty$: 
\begin{center}
	\begin{tabular}{|l|l|}
		\hline 
		\text{Noise assumption} &\text{ESPRIT Error} \\ \hline
		$\|\eta\|_\infty \lesssim {\rm SRF}^{-(2\lambda-2)}$ &$
		{\rm md}(\Omega,\hat \Omega)
		\lesssim  {\rm SRF}^{3\lambda-3} \|\eta\|_\infty $ \\ \hline
		$\|\eta\|_\infty \lesssim {\rm SRF}^{-(4\lambda-3)}/M$ &${\rm md}(\Omega,\hat \Omega) 
		\lesssim {\rm SRF}^{2\lambda-2}  \|\eta\|_\infty$ \\ \hline
	\end{tabular}
\end{center}

\commentout{

\begin{align*}
\text{(a) Moderate noise: } 
\|\eta\|_\infty \lesssim  {\rm SRF}^{-(2\lambda-2)} \quad 
\text{ESPRIT error: } 
{\rm md}(\Omega,\hat \Omega)
&\lesssim  {\rm SRF}^{3\lambda-3} \|\eta\|_{\infty} , \\
\text{(b) Small noise: }
\|\eta\|_\infty \lesssim   {\rm SRF}^{-(4\lambda-3)}/ M 
\quad 
\text{ESPRIT error: } 
{\rm md}(\Omega,\hat \Omega) 
&\lesssim {\rm SRF}^{2\lambda-2}  \|\eta\|_{\infty}.
\end{align*}
}
It was shown in \cite{batenkov2019super} that, for a slightly different model with continuous Fourier measurements, the min-max error for estimating the support is in the order of ${\rm SRF}^{-(2\lambda -2)}\|\eta\|_\infty/M$. There is evidence \cite{donoho1992superresolution,demanet2015recoverability,li2017stable,batenkov2018conditioning,batenkov2019super} showing that the min-max rates with continuous or discrete Fourier measurements are of the same order, so it is reasonable to assert that ${\rm SRF}^{-(2\lambda -2)}\|\eta\|_\infty/M$ is the min-max rate for our model. Our main result shows that ESPRIT achieves the min-max rate up to a factor of $M$ when noise is sufficiently small, i.e. $\|\eta\|_\infty \lesssim   {\rm SRF}^{-(4\lambda-3)}/ M$. 

If $\eta$ is independent Gaussian noise, i.e., $\eta \sim \mathcal{N}(0,\sigma^2 I)$, $\|\calH(\eta)\|_2$ satisfies the following concentration inequality \cite[Theorem 4]{liao2015multi}:
\begin{proposition}
\label{propnoise}
	If $\eta \sim \mathcal{N}(0,\sigma^2 I)$, then for any $t>0$,
	\label{lemmanoise}
	\begin{align*}
	\mathbb{E} \|\calH(\eta)\|_2 
	&\le \sigma \sqrt{2 \max(L+1,M-L+1)\log(M+2)},
	\\
	\mathbb{P}\left\{ 
	\|\calH(\eta)\|_2 \ge t
	\right\}
	&\le (M+2) \exp\left(-\frac{t^2}{2\sigma^2 \max(L+1,M-L+1)}\right). 
	\end{align*}  
\end{proposition}

In the case of Gaussian noise, $\mathbb{E} \|\calH(\eta)\|_2 $ scales like $\sigma \sqrt{M \log M}$, and $\|\eta\|_\infty$ scales like $\sqrt{\log M}\sigma$.
In the small noise regime where $\sigma \sqrt{M \log M} \propto {\rm SRF}^{-(4\lambda -3)}$, the ESPRIT error becomes 
\begin{equation}
\label{espritres2}
{\rm md}(\Omega,\hat \Omega) 
\propto \sqrt{\frac{\log M}{ M}}  {\rm SRF}^{2\lambda-2}   \frac{\sigma}{x_{\min}}
\quad \text{ or } \quad 
{\rm md}(\Omega,\hat \Omega) 
\propto \frac{{\rm SRF}^{2\lambda-2}}{x_{\min} \sqrt M }  \|\eta\|_\infty.
\end{equation}
In the case of Gaussian noise, our result matches 
the min-max rate up to a factor of $\sqrt{M}$ when noise is sufficiently small.

\subsection{ESPRIT in the well-separated case}

In this section, we derive stability bounds for ESPRIT when $\Delta\geq C/M$ for a reasonable constant $C>1$, we refer to as the {\it well-separated} case. The stability of ESPRIT in this regime is an easy consequence of the machinery we have developed so far. The key result that we employ is the inequality from \cite{moitra2015matrixpencil},
\begin{equation}
\label{eq:moitra}
\frac{C-1}{C}\, M \leq \sigma_{S}^2(\Phi_M(\Omega)),
\end{equation}
which holds under the assumption that $\Delta(\Omega)\geq C/M$ for some $C>1$. This was derived by using properties of the Beurling-Selberg majorant function, see \cite{vaaler1985some}. Combining this inequality with our machinery gives us the following result, and the proof is in Appendix \ref{seclemmas}. 

\begin{theorem}
	\label{thm:sep}
	Fix positive integers $M$ and $S$ such that $M\geq 4S$ is even and set $L= M/2$. For any $\mu\in\calM_S$, any minimum separation $\Delta=\Delta(\Omega)\geq 2C/M=C/L$ for some $C>2$, and any noise $\eta\in\C^M$ such that 
	\begin{equation}
	\label{eq:noise}
	\|\calH(\eta)\|_2
	\leq \frac{x_{\min} L}{4\sqrt{2S}} \frac{C-1}{C} \sqrt{1-\frac{C}{C-1}\frac{S}{L}} .
	\end{equation}
	then we have
	\[
	\md(\Omega,\hat\Omega)
	\leq \frac{20 S^2}{x_{\min}}\frac{\sqrt{L+1}}{L^{3/2}} \(\frac{C}{C-1}\)^{3/2} \(1-\frac{C}{C-1}\frac{S}{L}\)^{-1} \|\calH(\eta)\|_2.
	\] 
\end{theorem}

Let us explain what the inequality in the Theorem means in terms of the number of Fourier samples $M$. We suppress all of the terms that are independent of $M=2L$ and the noise term, the above inequality is of the form
\[
\md(\Omega,\hat\Omega)
\lesssim \frac{\|\calH(\eta)\|_2}{M}. 
\]

\section{Numerical simulations}
\label{secnum}

We next perform numerical simulations to verify Theorem \ref{thm1} and the scaling law in \eqref{espritres2} that was predicted by Theorem \ref{thmesprit}. 
In our simulations, the true support $\Omega$ contains $1,2,3$ or $4$ clumps ($A=1,2,3,4$) of $\lambda$ equally spaced objects consecutively spaced by $\Delta$, while the clusters are separated at least by $\beta/M$ with $\beta \ge 10$ (see Figure \ref{Fig_TransitionImage} (a) for an example). The coefficients $\{x_j\}_{j=1}^S$ have unit magnitudes and random phases.
We set $M=100, L = M/2=50$ and let $\Delta$ vary so that $\OmegaRF$ varies. 
Noise is gaussian: $\eta \sim \mathcal{N}(0,\sigma^2 I)$. 

The support error is measured by the matching distance ${\rm md}(\Omega,\hat\Omega) $. For each parameter setting, we randomly choose the phases of $x$, and run the experiments $100$ times with random noises and the fixed amplitudes $x$. The average support error for this $x$ is taken as the average of support matching distance within these $100$ experiments. In order to test ESPRIT's capability of dealing with arbitrary complex phases, we then take the worst average support error over $10$ random phases.

\subsection{Matching distance versus $\sigma_{\min}(\Phi_{L})$ and SRF}

Our first set of experiments is to verify Theorem \ref{thm1} (b), which proves the scaling law
\begin{equation}
{\rm md}(\Omega,\widehat\Omega) \propto \frac{\rm noise}{\sigma^2_S(\Phi_{M/2})}.
\label{thm1simulation}
\end{equation}
By employing Theorem \ref{thmsin}, the above can be rewritten as 
\begin{equation}
\label{thmsinsimulation2}
{\rm md}(\Omega,\widehat\Omega) \propto {\rm noise}\cdot {\rm SRF}^{2(\lambda-1)}.
\end{equation}
The noise noise level $\sigma$ is fixed in this experiment. For each fixed $\sigma$, we let ${\rm SRF}$ vary and record the average ${\rm md}(\Omega,\widehat\Omega)$ over $100$ experiments of random noises, for the worst random phases of $x$. Figure \ref{Fig_Thm1} displays the log-log plot of the average ${\rm md}(\Omega,\widehat\Omega)$ versus $\sigma_{\min}(\Phi_{M/2})$ and SRF for (a) $A=1$ and $\lambda = 2$ and (b) $A=2$ and $\lambda = 3$. The curves appear to be straight lines and the slopes of these curves verifies the theoretical prediction given by the scaling laws \eqref{thm1simulation} and \eqref{thmsinsimulation2}.

\begin{figure}[h]
	\centering
	\subfigure[$\lambda = 2$]{
		\includegraphics[width=7cm]{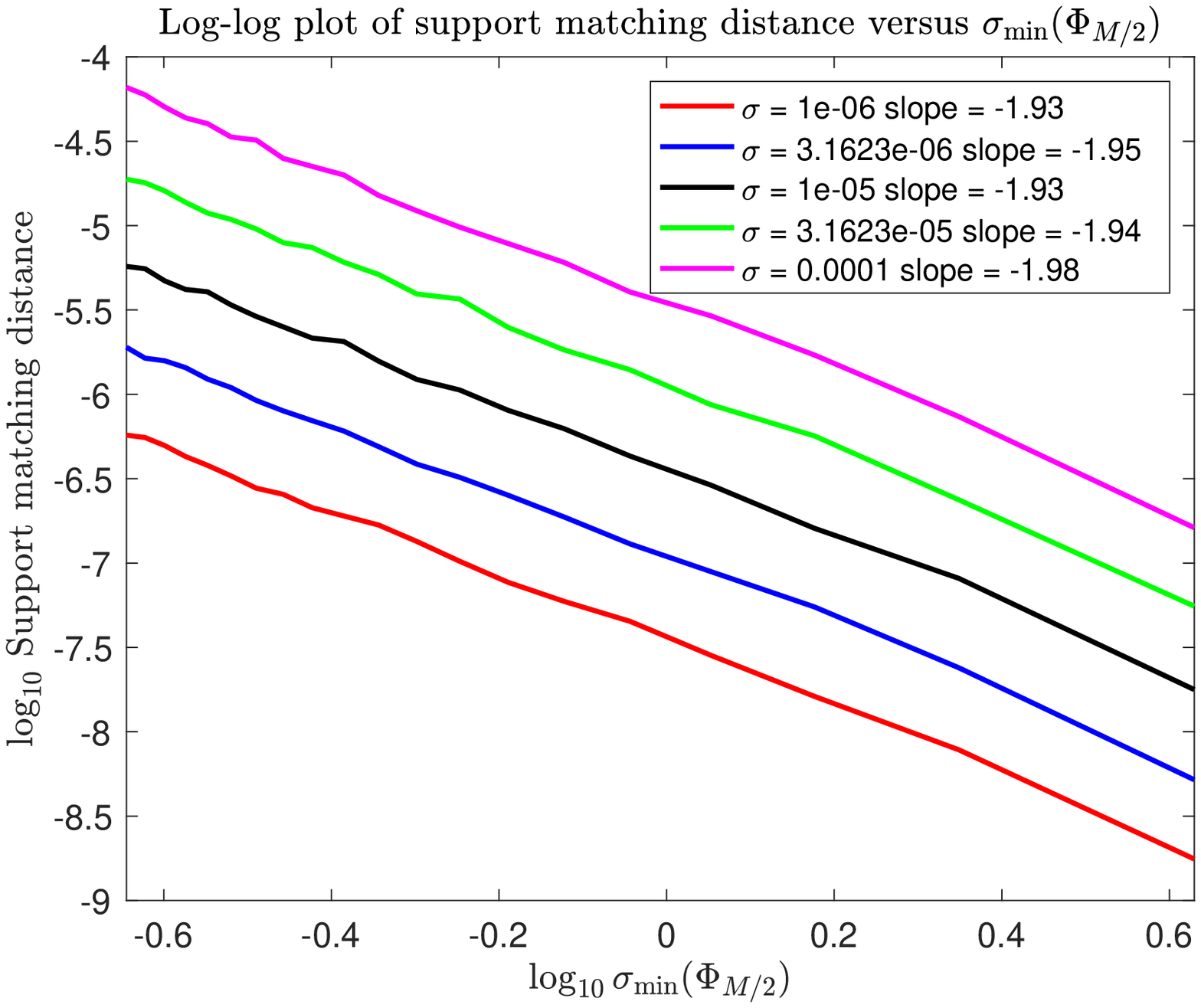}
		\includegraphics[width=7cm]{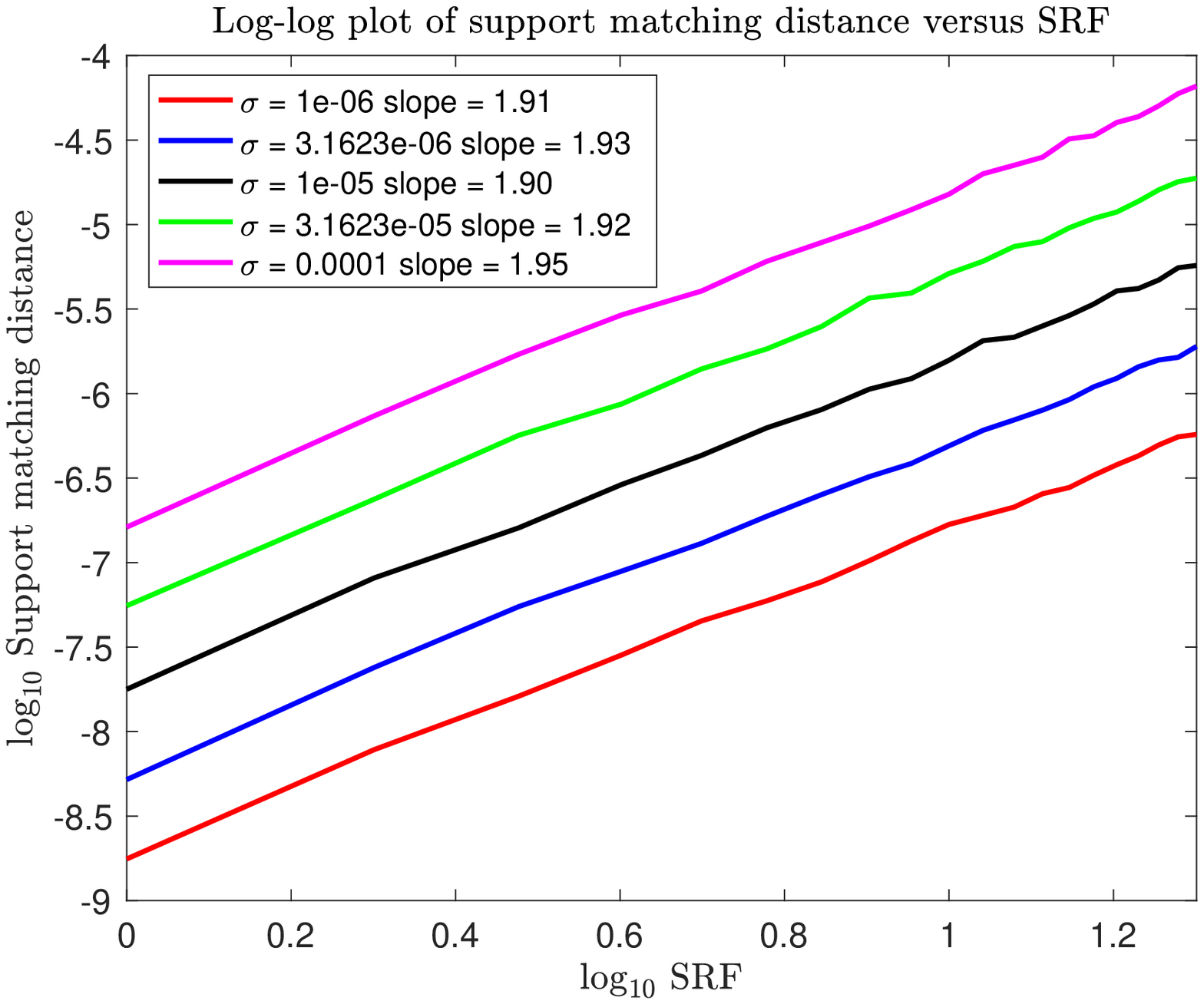}
	}

	\subfigure[$\lambda = 3$]{
		\includegraphics[width=7cm]{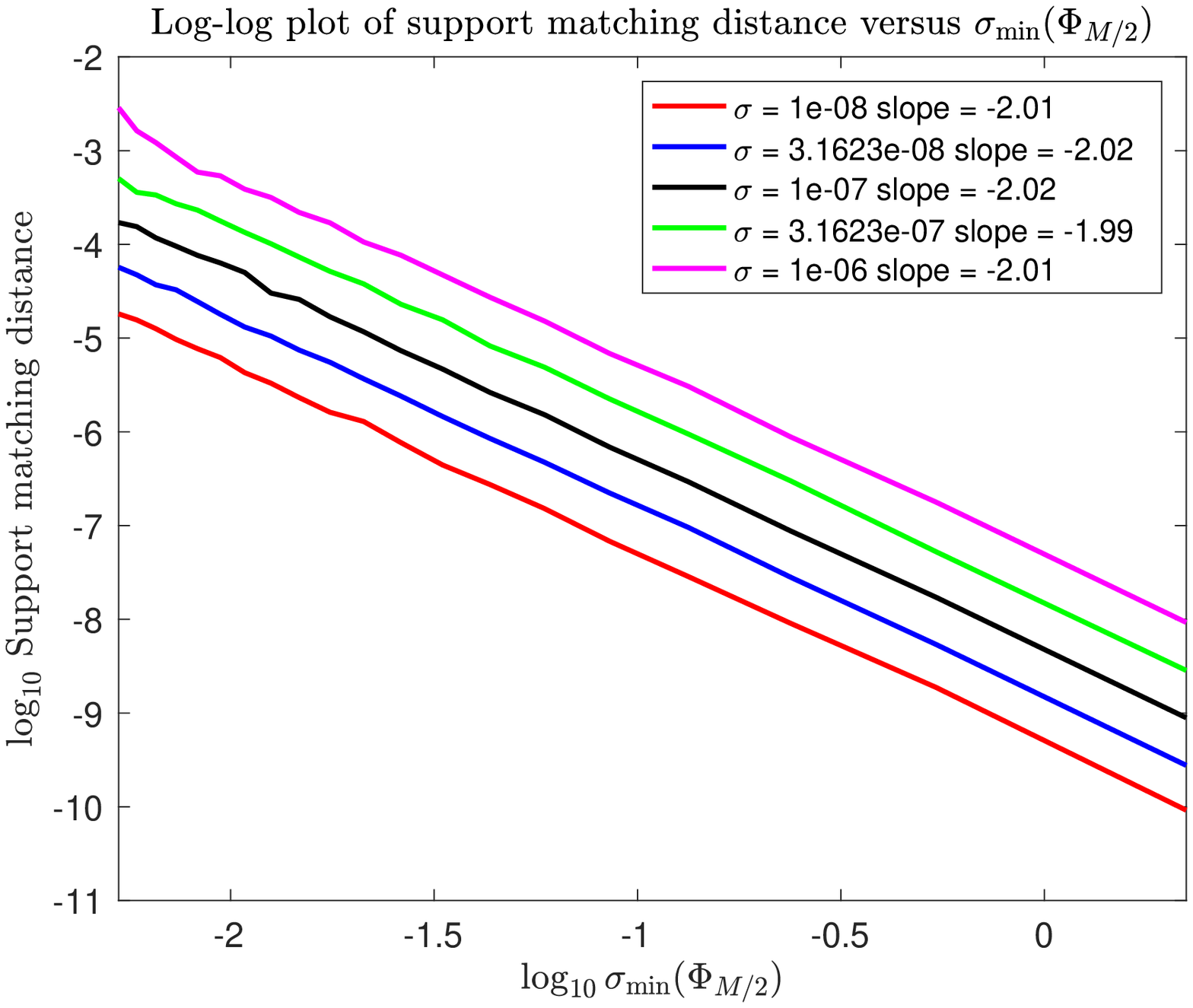}
		\includegraphics[width=7cm]{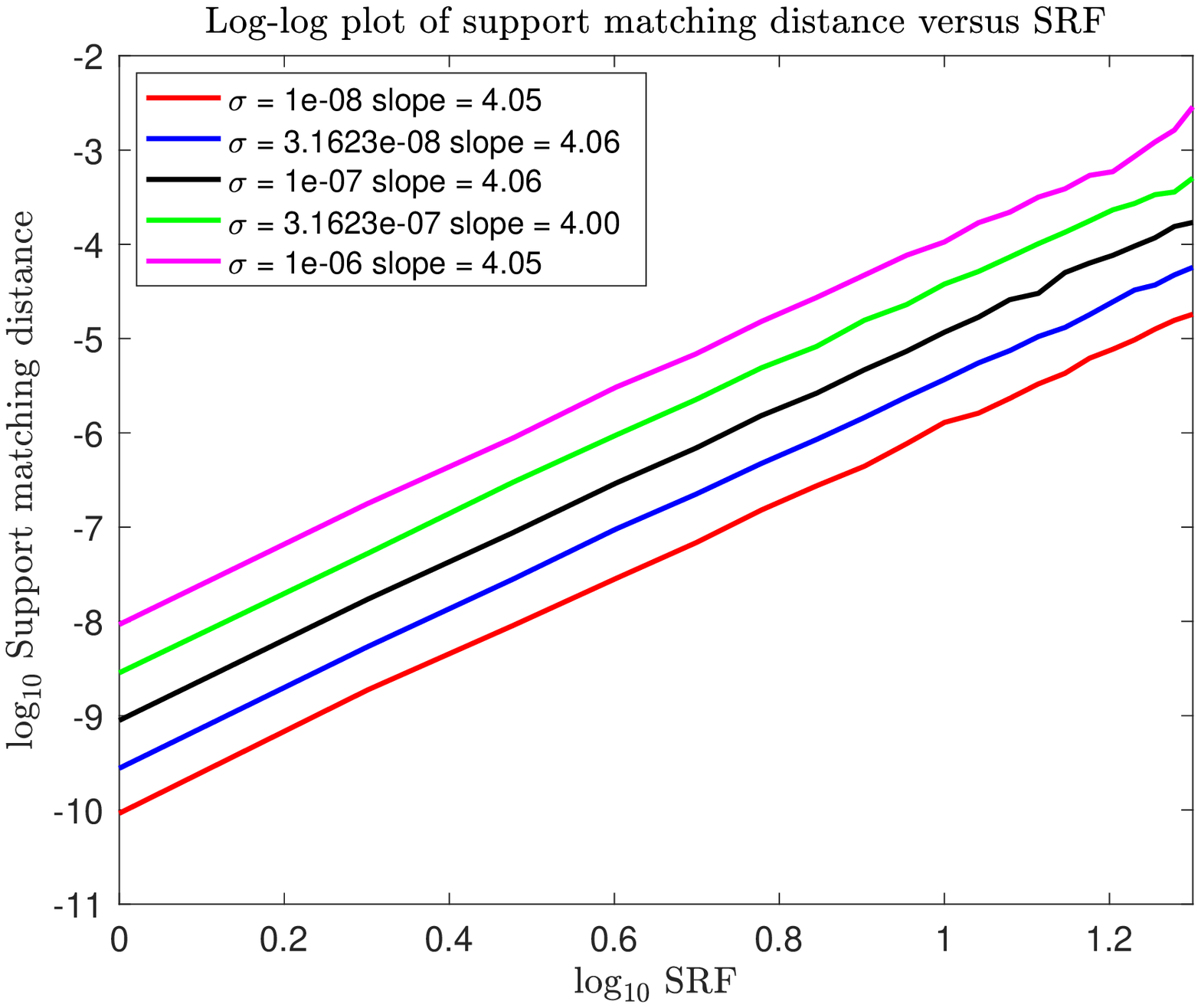}
	}
	\caption{These figures display the log-log plot of the average ${\rm md}(\Omega,\widehat\Omega)$ versus $\sigma_{\min}(\Phi_{M/2})$ (the left column) and SRF (the right column) for $\lambda = 2$ in (a) and $\lambda = 3$ in (b).}
	\label{Fig_Thm1}
\end{figure}

\subsection{Phase transition}

Our second set of experiments is more comprehensive than our first set, because we allow both ${\rm SRF}$ and $\sigma$ to vary. We perform $100$ trials for each SRF and $\sigma$, and we include the performance of MUSIC to serve as a comparison.

\begin{figure}[h]
	\centering
	\hspace{-0.5cm}
	\subfigure[$2$ clumps of $2$ objects 
	]
	{
		\begin{tikzpicture}[xscale = 0.8,yscale = 1]
		\draw[thick] (-5.5,0) -- (-0.1,0);
		\draw[red,thick,->] (-4.7,0) -- (-4.7,2);
		\draw[red,thick,->] (-4.4,0) -- (-4.4,2);
		\draw[blue,thick,<->] (-4.7,-0.2) -- (-4.4,-0.2);
		\node[blue,below] at (-4.1,-0.2) {$ \alpha/M$};
		\draw[blue,thick,<->] (-4.4,1) -- (-1,1);
		\node[blue,above] at (-2.8,1) {at least $\beta/M$};
		\draw[red,thick,->] (-1,0) -- (-1,2);
		\draw[red,thick,->] (-0.7,0) -- (-0.7,2);
		\draw[blue,thick,<->] (-1,-0.2) -- (-0.7,-0.2);
		\node[blue,below] at (-0.7,-0.2) {$ \alpha/M$};
		\node[above] at (-4.6,2.2) {$\Lambda_1$};
		\node[above] at (-0.75,2.2) {$\Lambda_2$};
		\node[below] at (-2.5,-1) {$\lambda = \max(|\Lambda_1|,|\Lambda_2|) = 2$};
		\end{tikzpicture}
	}
	\subfigure[MUSIC]{
		\includegraphics[width=5.7cm]{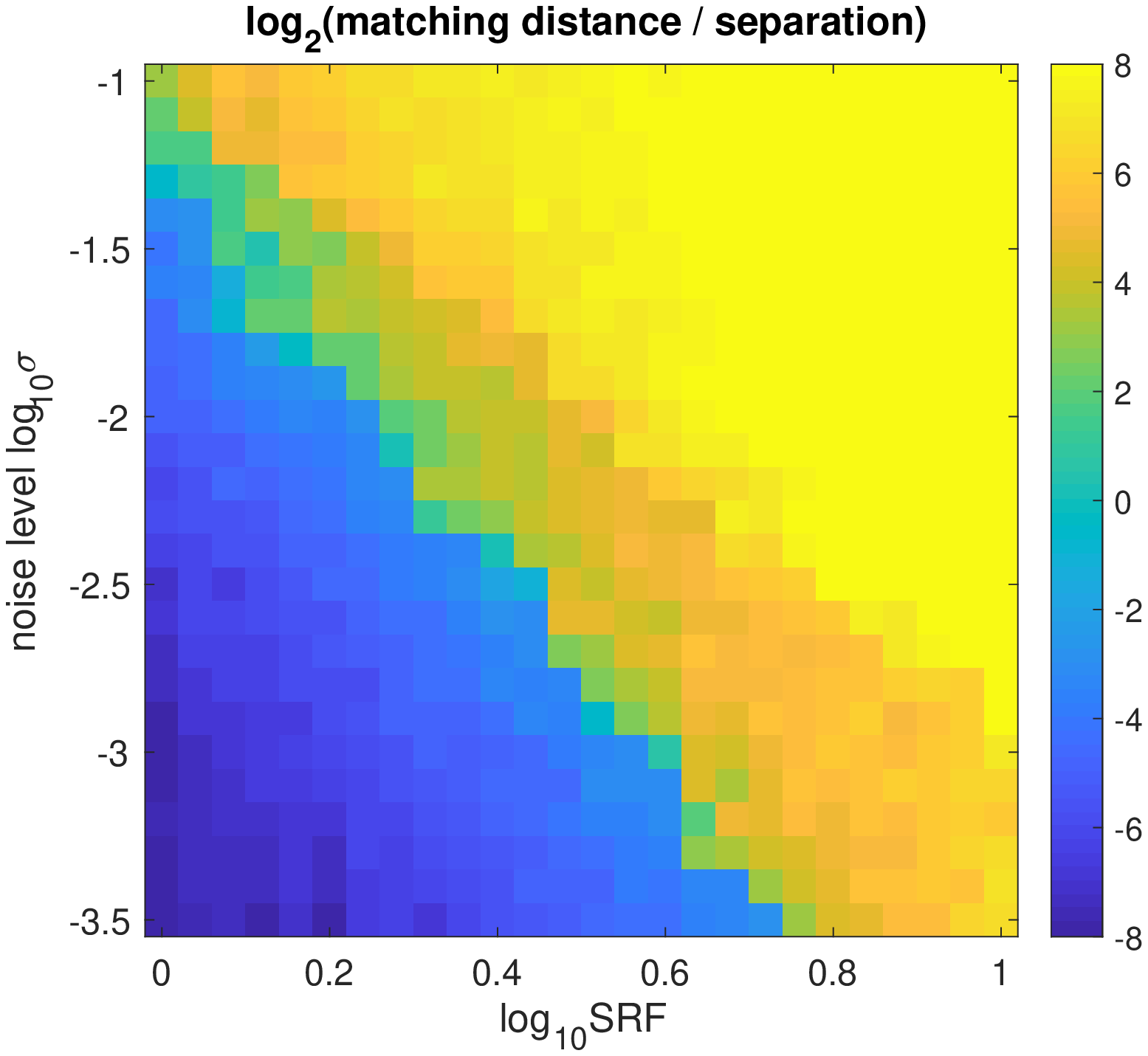}
	}
	\hspace{-1.4cm}
	\subfigure[ESPRIT]{
		\includegraphics[width=5.7cm]{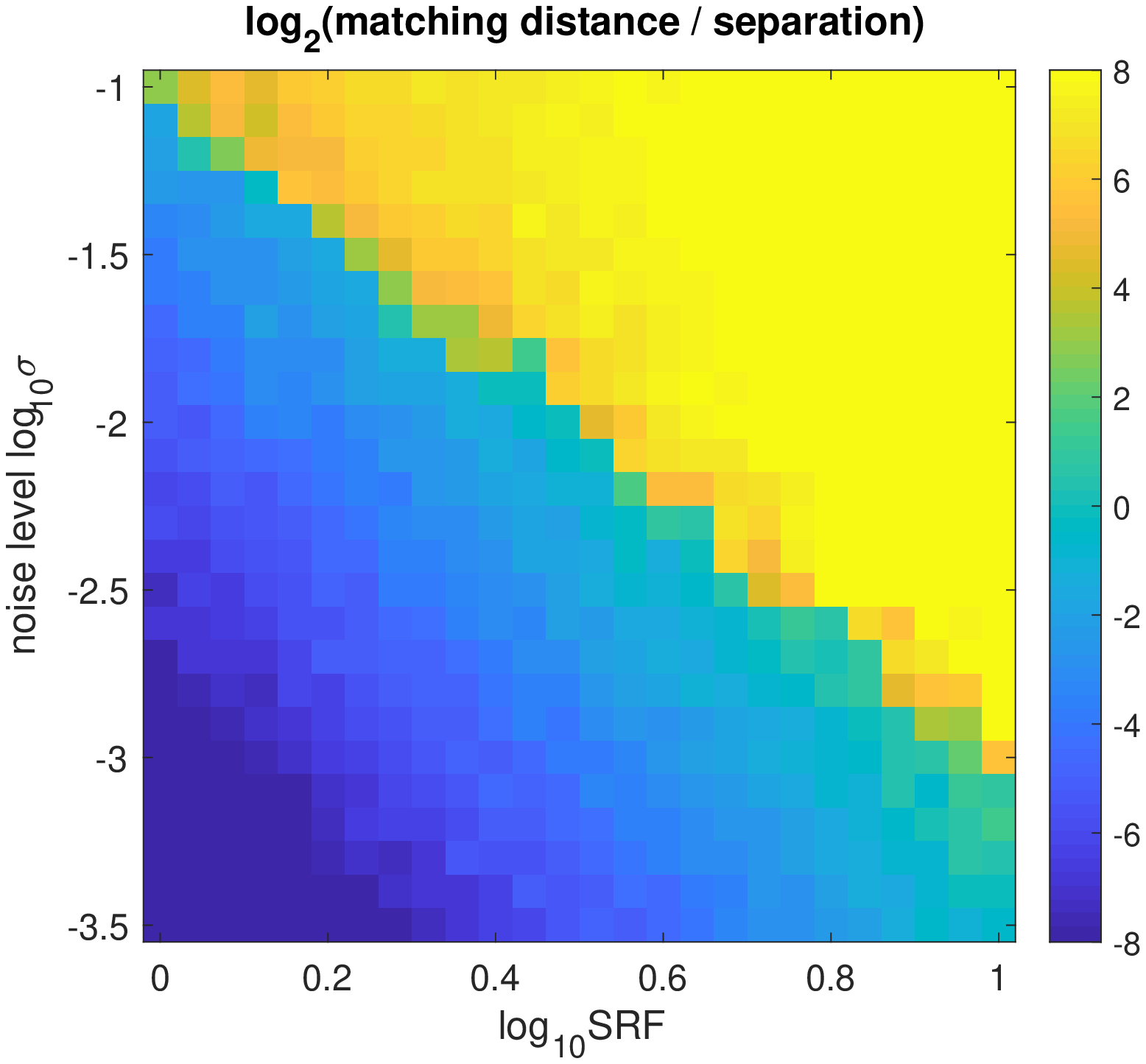}}
	\caption{(b) and (c) displays the average $\log_2 [{{\rm md}(\Omega,\hat\Omega)}/{\Delta}]$ over $100$ trials with respect to $\log_{10}\OmegaRF$ (x-axis) and $\log_{10}\sigma$ (y-axis) when $\Omega$ contains $2$ clusters of $2$ consecutively spaced objects: $A = 2$ and $\lambda = 2$.}
	\label{Fig_TransitionImage}
\end{figure}

Figure \ref{Fig_TransitionImage} (b) and (c) display the average value of $\log_2 [{\rm md}(\Omega,\hat\Omega)/{\Delta}]$ over $100$ trials with respect to $\log_{10}\OmegaRF$ (x-axis) and $\log_{10}\sigma$ (y-axis) when $\Omega$ contains $2$ clumps of $3$ consecutively spaced atoms: $A=2$ and $\lambda =2$. A clear phase transition demonstrates that MUSIC and ESPRIT are capable of resolving closely spaced \textit{complex-valued} objects as long as $\sigma$ is below certain threshold depending on $\OmegaRF$. ESPRIT outperforms MUSIC as it can tolerate a larger amount of noise.

In Figure \ref{Fig_PhaseTransition}, we display the phase transition curves. We say the output $\hat\Omega$ of either MUSIC or ESPRIT is successful if ${\rm md}(\Omega,\hat\Omega) \le {\Delta}/2$. The phase transition curves are extracted such that the success probability within $100$ simulations is above $95\%$.
Figure \ref{Fig_PhaseTransition} displays the phase transition curves of MUSIC and ESPRIT with respect to $\log_{10}\OmegaRF$ (x-axis) and $\log_{10}\sigma$ (y-axis). It appears that all phase transition curves are almost straight lines, manifesting that the noise level $\sigma$ that MUSIC and ESPRIT can tolerate satisfies 
\begin{equation}
\sigma \sim \OmegaRF^{-q(\Omega)}.
\label{eqmusicreslimit}
\end{equation}
The results of this experiment are consistent with the theoretical scaling laws  \eqref{espritres2}, which predicts that $q(\Omega) = 2\lambda -2$ independent of the clump number $A$.

We perform a least squares fitting of the curves by straight lines to obtain an empirical value of the exponent $q(\Omega)$, which is summarized in Table \ref{Table}. The numerical exponents of ESPRIT more or less match our theoretical estimation in \eqref{espritres2}. Our findings also indicate that ESPRIT is more robust to random noise than MUSIC. 

\begin{figure}[h!]
	\centering
	\subfigure[1 clump: $A=1$]{
		\includegraphics[width=7.8cm]{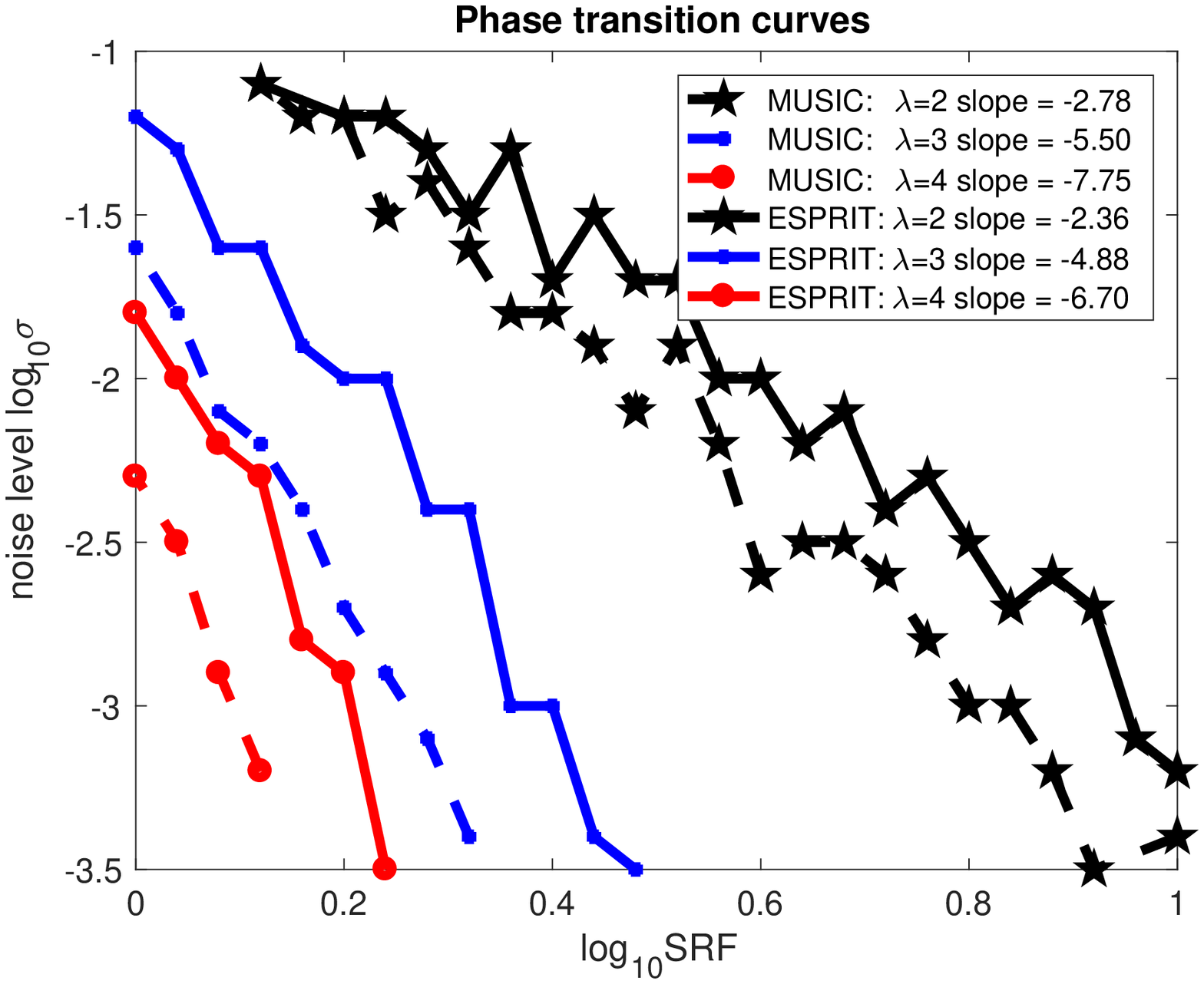}
	}
	\hspace{-1.2cm}
	\subfigure[2 clumps: $A=2$]{
		\includegraphics[width=7.8cm]{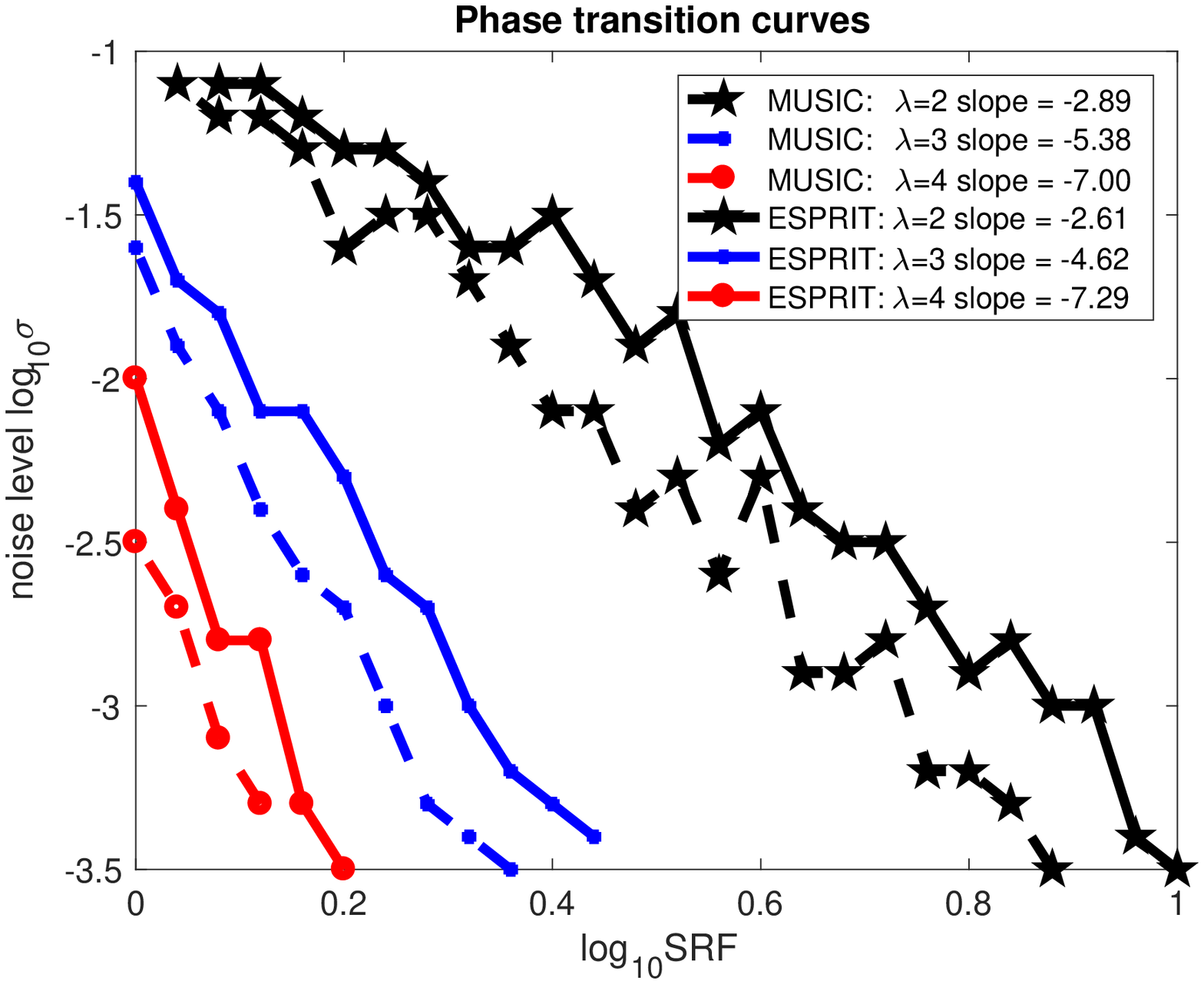}
	}
	\subfigure[3 clumps: $A=3$]{
		\includegraphics[width=7.8cm]{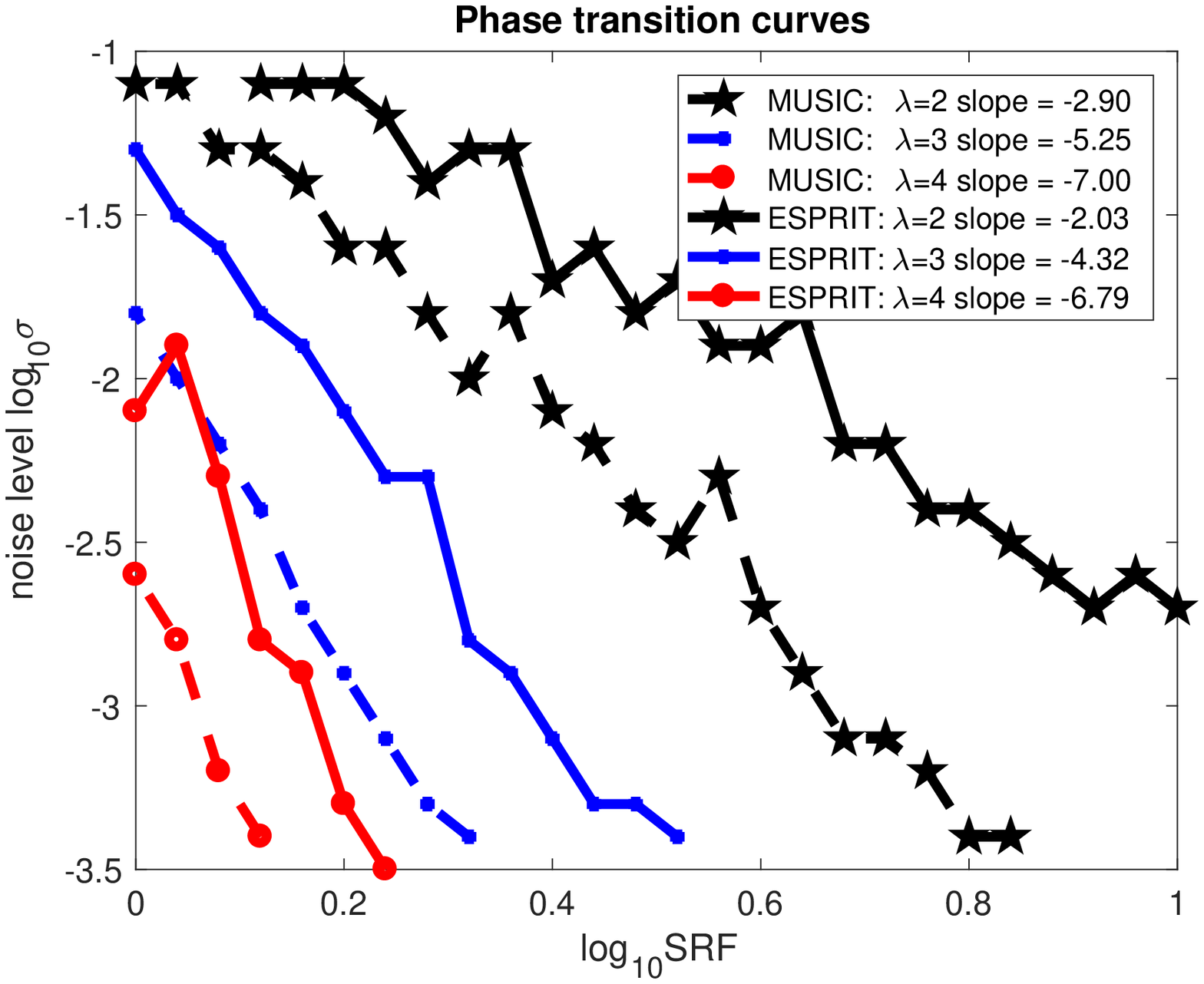}
	}
	\hspace{-1.2cm}
	\subfigure[4 clumps: $A=4$]{
		\includegraphics[width=7.8cm]{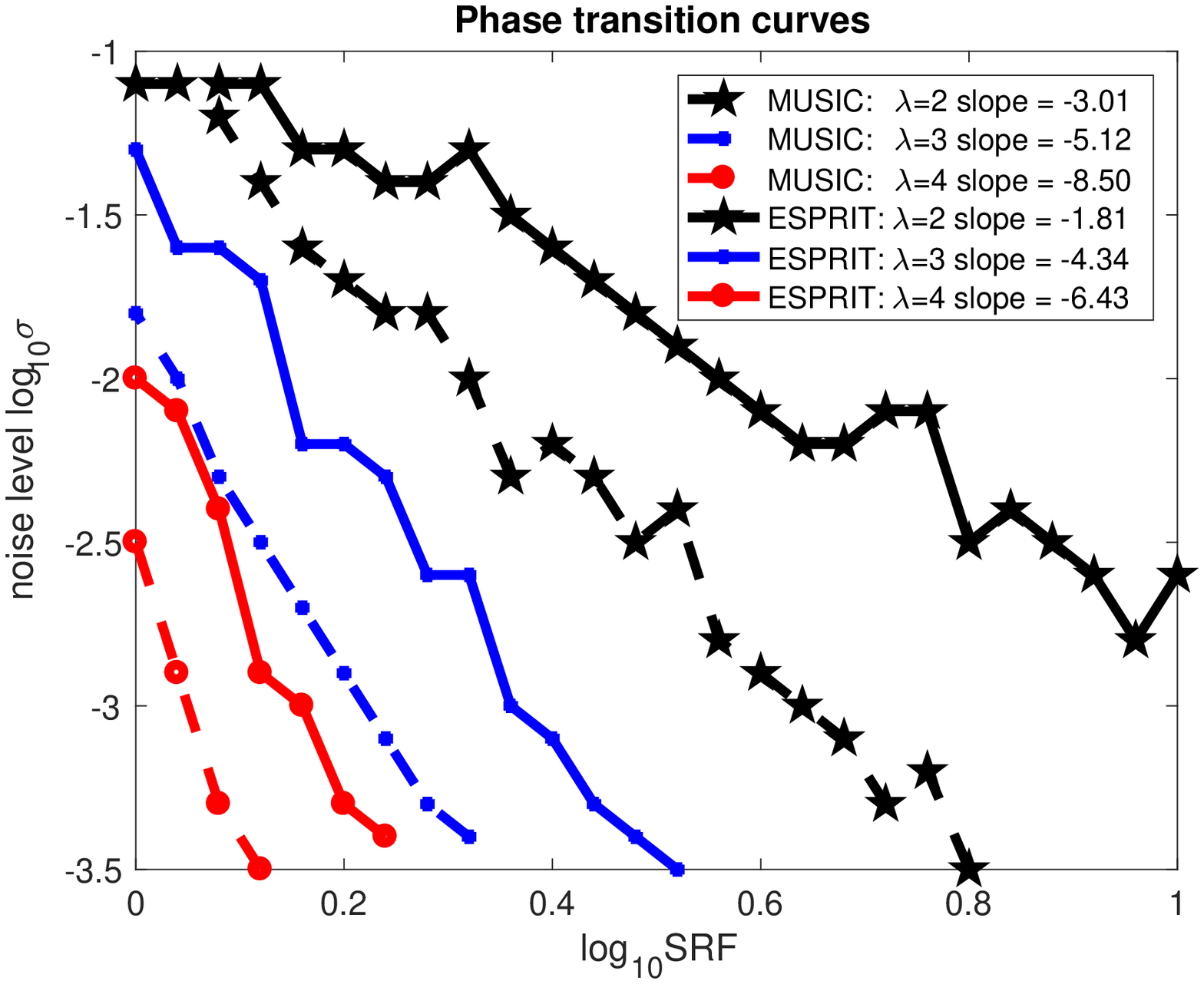}
	}
	\caption{The phase transition curves below which the success probability is at least $95\%$ for $\lambda=2,3,4$ with respect to $\log_{10}\OmegaRF$ (x-axis) and $\log_{10}\sigma$ (y-axis). The slopes are computed by least squares. }
	\label{Fig_PhaseTransition}
\end{figure}

	\begin{table}[h]
		\centering
		\begin{tabular}{ | c | c | c  | c | c|  c|}
			\hline			
			& $\lambda = 2$ & $\lambda = 3$ & $\lambda = 4$  & Numerical $q(\Omega)$ &Theoretical $q(\Omega)$  \\
			\hline
			1-clump: MUSIC  & $2.78$ & $5.50$ & $7.75$ & $2.49\lambda-2.11$ & $2\lambda-2$ \\
			\hline
			2-clump: MUSIC  & $2.89$ & $5.38$ & $7.00$ &  $2.06\lambda-1.08$ & $2\lambda-2$ \\
			\hline
			3-clump: MUSIC  & $2.90$ & $5.25$ & $7.00$ & $2.05\lambda-1.10$ & $2\lambda-2$ \\
			\hline
			4-clump: MUSIC  & $3.01$ & $5.12$ & $8.50$  &$2.75\lambda-2.70$  & $2\lambda-2$ \\
			\hline
			1-clump ESPRIT  & $2.36$ & $4.88$ & $6.70$ &$2.17\lambda-1.86$  & $2\lambda-2$ \\
			\hline
			2-clump ESPRIT  & $2.61$ & $4.62$ & $7.29$ &$2.34\lambda-2.18$  & $2\lambda-2$ \\
			\hline
			3-clump ESPRIT  & $2.03$ & $4.32$ & $6.79$  &$2.38\lambda-2.76$  & $2\lambda-2$ \\
			\hline
			4-clump ESPRIT  & $1.81$ & $4.34$ & $6.43$ &$2.31\lambda-2.74$  & $2\lambda-2$ \\
			\hline  
		\end{tabular}
		\caption{Numerical simulations of $q(\Omega)$ on the phase transition curves of MUSIC and ESPRIT.}
		\label{Table}
	\end{table}

\section{A non-harmonic uncertainty principle}
\label{secuncertainty}

Informally speaking, the uncertainty principle in Fourier analysis states that a function cannot be simultaneously localized in both time and frequency. This intuition was first formalized by the classical Heisenberg uncertainty principle, which showed that a square integrable function on $\R$ cannot simultaneously have small variance in both space and frequency. Many intriguing papers have provided inequalities of similar spirit for Fourier operators and spectral decompositions in a variety of settings. We refer the reader to \cite{folland1997uncertainty,benedetto2003weighted} and references therein for an overview of this subject.

We establish an uncertainty principle, but our main contributions in this area appear to be non-standard in two ways: the statements are not given in terms of norms and we treat discrete non-harmonic Fourier series. More precisely, we are mainly interested in the set $\calM_S$ of all non-zero complex valued discrete measures $\mu$ on the torus $\T$ consisting of at most $S$ atoms. Any $\mu\in\calM_S$ has the representation $\mu=\sum_{j=1}^S u_j \delta_{\omega_j}$ for some $u_j\in\C$ and $\omega_j\in\T$. Recall that $\hat\mu$ denotes the Fourier transform of $\mu$ and it is convenient to define the quantity 
\begin{equation}
\label{def:norm}
\|\hat\mu\|_{\ell^2_N}
=\(\sum_{k=0}^{N-1} |\hat\mu(k)|^2\)^{1/2}. 
\end{equation}
This section is primarily concerned with obtaining an upper bound on the quantity,
\begin{equation}
\label{eq:C}
C_{N,S}
:=\sup_{\mu\in\calM_S} C_N(\mu)
\quad\text{where}\quad
C_N(\mu)
:=\frac{|\hat\mu(0)|}{\|\hat\mu\|_{\ell^2_N}},
\end{equation}
which is well-defined for $N\geq S$ in view of the following result. 

\begin{proposition}
	\label{prop0}
	For any $\mu\in\calM_S$ and $N\geq S$, we have $\|\hat\mu\|_{\ell^2_N}>0$.
\end{proposition}

\begin{proof} 
	If $\mu=\sum_{j=1}^S u_j\delta_{\omega_j}$ and $\Omega=\{\omega_j\}_{j=1}^S$, then the Vandermonde matrix $\Phi_{N-1}(\Omega)$ has full rank when $N\geq S$. If $\|\hat\mu\|_{\ell^2_N}=0$, then 
	\[
	0
	=\hat\mu(k)
	=\sum_{j=1}^S u_j e^{-2\pi ik \omega_j}
	\quad \text{for}\quad 0\leq k\leq N-1.
	\]
	This implies $u=0$, which contradicts the assumption that $\mu$ is non-zero.
\end{proof} 

We note that $C_{N,S}=\infty$ when $N<S$ because it is not hard to show there exists a $\mu\in\calM_S$ such that $\|\hat\mu\|_{\ell^2_N=0}$. We also have that $C_{S,S}=1$. To see this, we can consider a measure $\mu=\sum_{j=0}^{N-1} u_j\delta_{j/N}$, where we shall pick the amplitudes momentarily. Note that $\mu\in\calM_S$ and that the Fourier coefficients $\{\hat\mu(k)\}_{k=0}^{N-1}$ consists of the discrete Fourier transform of the vector $u\in\C^N$. Thus, we can pick $u\in\C^N$ such that its DFT is precisely the canonical basis vector $e_0$. Doing this provides an example of a measure $\mu\in\calM_S$ such that $C_N(\mu)=1$. The case where $S=1$ and $N> 1$ is also trivial, since a direct computation shows that $C_{N,1}=1/\sqrt{N}$.

Thus, the only interesting case is when $N>S>1$. Obviously we have the trivial upper bound $C_{N,S}\leq 1$, and the point of the below results is prove a better estimate by using the assumption that $\mu\in\calM_S$. The quantity $C_N(\mu)$ describes the concentration or localization of $\hat\mu$ in its zero-th Fourier coefficient. However, we suspect that $C_N(\mu)<1$ when $N>S$ because $\mu$ is supported in $S$ points and it is hard to imagine that $\hat\mu$ would be supported in exactly 1 out of $N>S$ frequencies. To provide some support for this claim, we first generalize a result of Donoho-Stark \cite[Theorem 1]{donoho1989uncertainty} which was originally proved for (harmonic) Fourier series.

\begin{proposition}
	\label{prop1}
	Let $\mu\in\calM_S$ and $N\geq S$. For any $N$ consecutive Fourier coefficients of $\mu$, at least $\lfloor N/S \rfloor$ of them are non-zero. Moreover, for any $\mu\in\calM_S$ and $N>S$, we have $C_N(\mu)<1$. 
\end{proposition}

\begin{proof}
	Fix a non-zero measure $\mu=\sum_{j=1}^S u_j\delta_{\omega_j}$, where $u\in\C^S$ and $\Omega=\{\omega_j\}\subset\T$. For any set containing $N$ consecutive integers, from it we extract $\lfloor N/S\rfloor$ disjoint subsets, where each set contains $S$ consecutive integers. Call one these sets $\{n,n+1,\dots,n+S-1\}$. Suppose for the purpose of contradiction that $\hat\mu(k)=0$ for each $n\leq k\leq n+S-1$. Then $\sum_{j=1}^S u_je^{-2\pi ik\omega_j}=0$ for each $n\leq k\leq n+S-1$. This is a system of equations, and since square Vandermonde matrices are invertible, this implies $u_j=0$ for $1\leq j\leq S$, which is a contradiction. Thus, there is at least one integer in $k\in \{n,n+1,\dots,n+S-1\}$ for which $\hat\mu(k)\not=0$. Repeating this argument for each of the $\lfloor N/S \rfloor$ sets proves the first statement of the proposition.
	
	To see why the second statement of the proposition follows, consider the set $\{1,2,\dots,S\}$. Then for any $\mu\in\calM_S$, there is a $k\in \{1,2,\dots,S\}$ such that $|\hat\mu(k)|>0$. This shows that $\|\hat\mu\|_{\ell^2_N}>|\hat\mu(0)|$ or equivalently, $C_N(\mu)<1$.  
\end{proof} 

Proposition \ref{prop1} and uncertainty principles of Donoho-Stark type, see \cite{matusiak2004donoho} for generalizations, estimate the number of non-zero Fourier coefficients, but do not say how large they must be. In contrast, the following theorem can be seen as a statement about the size of the amplitudes.

\begin{theorem}
	\label{thm:UPcomplex}
	If $N>S>1$, then $C_{N,S}\leq \sqrt{1-4^{-S}}$.
\end{theorem}

Before we prove the theorem, we first state the a simple and useful observation regarding the dual relationship between polynomial interpolation and Fourier transforms. 

\begin{proposition}
	\label{prop4}
	For any $\mu\in\calM_S$ supported in $\Omega$, if there exists a continuous function $f$ on $\T$ such that $f=1$ on $\Omega$ and $\hat f$ is supported in a set $\Lambda\subset\Z$, then 
	\[
	|\hat\mu(0)|
	\leq \|f\|_{L^2(\T)}  \(\sum_{k\in\Lambda} |\hat\mu(k)|^2\)^{1/2}. 
	\]
\end{proposition} 

\begin{proof}
	This is a basic consequence of duality and the Parseval theorem. Let $\mu=\sum_{j=1}^S u_j\delta_{\omega_j}$. If $f$ is continuous and satisfies the assumed properties, then we have 
	\[
	|\hat\mu(0)|
	=\Big|\sum_{j=1}^S u_j\Big|
	=\Big|\int_\T \overline{f}\ d\mu\Big| 
	=\Big|\sum_{k\in\Lambda} \hat f(k) \hat\mu(k) \Big|
	\leq \|f\|_{L^2(\T)} \(\sum_{k\in\Lambda} |\hat\mu(k)|^2\)^{1/2}. 
	\]
\end{proof}

\begin{proof}[Proof of Theorem \ref{thm:UPcomplex}]
	The main idea is to construct, for each set $\Omega$ of cardinality $S$, a continuous function $f_\Omega$ with the properties listed in Proposition \ref{prop4} such that  $\|f_\Omega\|_{L^2(\T)}$ is uniformly bounded in $\Omega$. To do this, it is simpler to construct a function that vanishes on $\Omega$ instead. 
	
	For now, we fix a set $\Omega=\{\omega_j\}_{j=1}^S\subset\T$ and consider the trigonometric polynomial,  
	\begin{equation}
	\label{eq:poly}
	p_\Omega(\omega)
	:=(-1)^{S}\ \prod_{j=1}^S\ (e^{2\pi i(\omega-\omega_j)}-1). 
	\end{equation}
	Its zero set is precisely $\{\omega_j\}_{j=1}^S$. To see why $\hat{p_\Omega}$ is supported in $\{0,1,\dots,S\}$, we could expand out the product in the definition of $p_\Omega$, which would provide its Fourier series representation. Doing this also shows that $\hat{p_\Omega}(0)=1$. It is natural to consider this function, since it is the polynomial with the minimum degree that generates the ideal of trigonometric polynomials vanishing on $\Omega$. 
	
	For any $\alpha>0$, we define the family of functions
	\[
	f_{\Omega,\alpha}(\omega)
	:=1-\alpha p_\Omega(\omega). 
	\]
	By construction, we have $f_{\Omega,\alpha}=1$ on $\Omega$, $\hat{f_{\Omega,\alpha}}$ is supported in the set $\{0,1,\dots,S\}$, $\hat{f_{\Omega,\alpha}}(0)=1-\alpha$, and  $\hat{f_{\Omega,\alpha}}(k)=-\alpha \hat{p_\Omega}(k)$ for all $k\not=0$. It follows from these properties, Parseval's theorem, and H\"older's inequality that
	\begin{align*}
	\|f_{\Omega,\alpha}\|_{L^2(\T)}^2
	&=\sum_{k=0}^S |\hat{f_{\Omega,\alpha}}(k)|^2 \\
	&=(1-\alpha)^2+ \alpha^2 \sum_{k=1}^S |\hat{p_\Omega}(k)|^2 \\
	&=(1-\alpha)^2-\alpha^2|\hat{p_\Omega}(0)|^2+\alpha^2 \sum_{k=0}^S |\hat{p_\Omega}(k)|^2 \\
	&\leq (1-\alpha)^2-\alpha^2+\alpha^2 \|p_\Omega\|_{L^2(\T)}^2\\
	&\leq 1-2\alpha+\alpha^2 \|p_\Omega\|_{L^\infty(\T)}^2. 
	\end{align*}
	The optimal $\alpha$ that minimizes the right hand side is $\alpha=\|p_\Omega\|_{L^\infty(\T)}^{-2}$. Let $f_\Omega$ be the function $f_{\Omega,\alpha}$ for this particular value of $\alpha$. This shows that
	\begin{equation}
	\label{eq:f1}
	\|f_\Omega\|_{L^2(\T)}^2
	\leq 1-\|p_\Omega\|_{L^\infty(\T)}^{-2}.
	\end{equation}
	All that remains is to upper bound $\|p_\Omega\|_{L^\infty(\T)}$ uniformly in $\Omega$. Since we do not want to use any information about $\Omega$ except for its cardinality, the only bound available to us is
	\begin{equation}
	\label{eq:f2}
	\|p_\Omega\|_{L^\infty(\T)}
	=\sup_{\omega\in\T}\  \prod_{j=1}^S |e^{2\pi i(\omega-\omega_j)}-1|
	\leq 2^S. 
	\end{equation}
	We are ready to apply Proposition \ref{prop4} with the set $\Lambda=\{0,1,\dots,S\}$. For each $\mu\in\calM_S$ supported in $\Omega$, we note that $f_\Omega$ satisfies the required properties. Together with inequalities \eqref{eq:f1} and \eqref{eq:f2} shows that
	\[
	|\hat\mu(0)|
	\leq \|f_\Omega\|_{L^2(\T)} \|\hat\mu\|_{\ell^2_{S+1}}
	\leq \sqrt{1-4^{-S}}\  \|\hat\mu\|_{\ell^2_{S+1}}. 
	\]
	This inequality is uniform over all sets $\Omega$ of cardinality $S$. Rearranging and taking the supremum over all $\mu\in\calM_S$ shows that
	\begin{equation}
	C_{N,S}
	=\sup_{\mu\in\calM_S} C_N(\mu)
	\leq \sup_{\mu\in\calM_S} C_{S+1}(\mu)
	\leq \sqrt{1-4^S}.
	\end{equation}
\end{proof}

At this point, we have proved the necessary uncertainty principle that is required for the stability analysis of ESPRIT. Since we have already introduced this problem, we make several further observations that might be of independent interest. 

We next show that the case of real measures appears to be different from that of the complex case. Let $\calM_{S,\R}$ denote the set of all non-zero real discrete measures supported on $\T$ consisting of at most $S$ atoms. We define the quantity
\begin{equation}
\label{eq:Creal}
C_{N,S,\R}
:=\sup_{\mu\in\calM_{S,\R}} C_N(\mu). 
\end{equation}
We obtain a bound for $C_{N,S,\R}$ that is significantly smaller than our bound for $C_{N,S}$. 

\begin{theorem}
	\label{thm:UPreal}
	If $N>S>1$, then $C_{N,S,\R}\leq \sqrt{1-(8S-1)^{-1}}$. 
\end{theorem}

\begin{proof}
	The strategy for the proof is the same as that of Theorem \ref{thm:UPcomplex}, but we have the additional advantage of working with non-negative polynomials instead, and we shall see why. Indeed, fix a real $\mu\in\calM_{S,\R}$. Then $\hat\mu(k)=\overline{\hat\mu(-k)}$ for each $k\in\Z$ and so 
	\begin{equation}
	\label{eq:mureal}
	\frac{|\hat\mu(0)|^2}{\sum_{k=0}^{N-1} |\hat\mu(k)|^2}
	=\frac{2|\hat\mu(0)|^2}{|\hat\mu(0)|^2+\sum_{k=-N+1}^{N-1}|\hat\mu(k)|^2}.
	\end{equation}
	It suffices to lower bound the denominator in terms of $|\hat\mu(0)|$. 
	
	For now, we fix a set $\Omega=\{\omega_j\}_{j=1}^S\subset\T$. Again, we let $p_\Omega$ be the trigonometric polynomial defined by \eqref{eq:poly}, and we define the functions
	\[
	f_\Omega(\omega)
	:=1-h_\Omega(\omega), 
	\quad\text{where}\quad
	h_\Omega(\omega):=\frac{|p_\Omega(\omega)|^2}{\|p_\Omega\|_{L^\infty(\T)}^2}. 
	\]
	Notice that $f_\Omega=1$ on $\Omega$, $\hat{f_\Omega}$ is real valued, has Fourier transform supported in $\{-S,\dots,S\}$, and $0\leq h_\Omega\leq 1$.
	
	By construction, $h_\Omega$ is a positive continuous function on $\T$ with $\|h_\Omega\|_{L^\infty(\T)}=1$. There exists a $\omega_0\in\T$ such that $h(\omega_0)=1$. Since $S>1$, $h_\Omega$ has at least two roots. The intermediate value theorem guarantees the existence of an interval $(\omega_-,\omega_+)$ containing $\omega_0$ such that $h(\omega_{\pm})=1/2$ and $h=|h|\geq 1/2$ on the interval $[\omega_-,\omega_+]$. We argue that $[\omega_-,\omega_+]$ cannot be too small. To see this, we define the axillary polynomial 
	\[
	H_\Omega(\omega)=2h_\Omega(\omega)-1. 
	\]
	Note that $\|H_\Omega\|_{L^\infty(\T)}=1$, $H_\Omega(x_{\pm})=0$ and $H_\Omega(x_0)=1$. We proceed to apply Tur\'an's theorem \cite{turan1946rational}, which we state below.
	
	\begin{theorem}
		Let $P$ be a non-trivial polynomial of degree $n$ such that
		$|P(1)| = \max_{|z|=1} |P(z)|$. Then for any root $w$ of $P$ on the unit circle, $|\arg(w)|\geq \pi/n$. Moreover, if $|\arg(w)|=\pi/n$, then $P(z) = c(1 + z^n)$ for some non-zero $c\in\C$. 
	\end{theorem}
	
	By construction, when extended to the complex plane, $H_\Omega$ is a polynomial of degree $2S$, has zeros at $e^{2\pi i \omega_+}$ and $e^{2\pi i \omega_-}$ on the unit complex circle, and the maximum of $H_\Omega$ on the unit circle is attained at $e^{2\pi i\omega_0}$. Tur\'an's theorem implies that $|\arg(e^{2\pi i (\omega_+-\omega_0)})|\geq \pi/(2S)$ and $|\arg(e^{2\pi i (\omega_0-\omega_-)})|\geq \pi/(2S)$. These inequalities imply 
	\begin{align*}
	\big|[\omega_-,\omega_+]\big|
	&\geq |\omega_+-\omega_0|_\T + |\omega_0-\omega_-|_\T \\
	&=\frac{1}{2\pi} \big(|\arg(e^{2\pi i (\omega_+-\omega_0)})|+|\arg(e^{2\pi i (\omega_0-\omega_-)})|\big) \geq \frac{1}{2S}.
	\end{align*}
	We could have reached a similar conclusion using the Bernstein inequality and the mean value theorem, but that argument would yield an inequality with a slightly worse constant. By construction, $h_\Omega\geq 1/2$ on $[\omega_-,\omega_+]$. This implies
	\[
	\int_{\omega_-}^{\omega_+} h_\Omega(\omega)\ d\omega
	\geq \frac{\big|[\omega_+,\omega_-] \big|}{2}
	\geq \frac{1}{4S}.
	\]
	
	We are ready to upper bound $\|f_\Omega\|_{L^2(\T)}$. Since $0\leq h_\Omega\leq 1$, we have
	\begin{align*}
	\|f_\Omega\|_{L^2(\T)}^2
	&\leq \int_{\T} (1-h_\Omega(\omega))\ d\omega
	\leq 1-\int_{\omega_-}^{\omega_+} h_\Omega(\omega)\ d\omega
	\leq 1-\frac{1}{4S}. 
	\end{align*}
	We apply Proposition \ref{prop4} for the set $\Lambda=\{-S,\dots,S\}$ and we note that $f_\Omega$ satisfies the required properties. Thus, for any $\mu\in\calM_S$ supported in $\Omega$, we have 
	\[
	|\hat\mu(0)|
	\leq \|f_\Omega\|_{L^2(\T)} \(\sum_{k=-S}^S |\hat\mu(k)|^2\)^{1/2}
	\leq \sqrt{\frac{4S-1}{4S}} \(\sum_{k=-S}^S |\hat\mu(k)|^2\)^{1/2}. 
	\]
	This shows that for any $\mu\in\calM_{S}$ supported in $\Omega$,
	\[ 
	\frac{2|\hat\mu(0)|^2}{|\hat\mu(0)|^2+\sum_{k=-L}^{L}|\hat\mu(k)|^2}
	\leq \frac{2}{1+4S/(4S-1)}
	\leq 1-\frac{1}{8S-1}.
	\]
	This bound is uniform over all $\Omega$ of cardinality $S$. This completes the proof of the theorem.
\end{proof}

The upper bounds in Theorems \ref{thm:UPcomplex} and \ref{thm:UPreal} are slightly less than 1, which is to be expected for an estimate that is uniform over all $\mu\in\calM_S$. If we impose additional restrictions on the support of $\mu$, then we obtain a much better bound on $C_N(\mu)$. Recall that the minimum separation $\Delta$ of a set $\Omega=\{\omega_j\}_{j=1}^S$ is defined in equation \eqref{eq:minsep}. 

\begin{proposition}
	\label{prop:UP2}
	If $N>S$, $\mu\in\calM_S$ is supported in $\Omega$, and $\Delta\geq C/(N-1)$ for some $C>1$, then 
	\[
	C_N(\mu)\leq \min\(1,\sqrt{\frac{C}{C-1}}\sqrt{\frac{S}{N-1}}\). 
	\]
\end{proposition} 

\begin{proof}
	Fix any $\mu=\sum_{j=1}^S u_j\delta_{\omega_j}$ satisfying the assumption that $\Delta \geq C/(N-1)$ and let $\Phi_{N-1}=\Phi_{N-1}(\Omega)$ where $\Omega=\{\omega_j\}_{j=1}^S$. The paper \cite{moitra2015matrixpencil} showed that $\Phi_{N-1}$ is well-conditioned and provided a lower bound on the operator norm of $\Phi_{N-1}$. By this result, we get
	\begin{equation}
	\label{eq:2bound}
	\|\hat\mu\|_{\ell^2_N}
	=\|\Phi_{N-1}u\|_2
	\geq \sigma_{S}(\Phi_{N-1})\|u\|_2
	\geq \sqrt{(1-C^{-1})(N-1)} \|u\|_2.
	\end{equation}
	Using this inequality and Cauchy-Schwarz, we obtain
	\[
	C_N(\mu)
	=\frac{|\hat\mu(0)|}{\|\hat\mu\|_{\ell^2_N}}
	\leq \sqrt{\frac{S}{(1-C^{-1})(N-1)}}. 
	\]
	We also have the trivial inequality $C_N(\mu)\leq 1$.
\end{proof}

The technique in Proposition \ref{prop:UP2} has its limitations. The key step is \eqref{eq:2bound}, which controls $\|\hat\mu\|_{\ell^2_N}$ via  $\sigma_{S}(\Phi_{N-1})$, which is not uniformly bounded in $\Omega$. Hence, this argument is too wasteful and cannot be used to deduce Theorem \ref{thm:UPcomplex} where no assumptions on $\Omega$ are placed. In fact, this shows that it is impossible to obtain the theorem using an inequality that estimates $\|\Phi_{N-1} u\|_2$ in terms of $\|u\|_2$. In particular, large sieve inequalities, see \cite{montgomery1978analytic} for an overview, do not appear to be helpful for this problem.

The upper bounds given in Theorem \ref{thm:UPcomplex} and \ref{thm:UPreal} are independent of $N$, and the reader might wonder if they can be improved in the regime where $N$ is significantly larger than $S$. It is straightforward to see that the most optimistic decay we could expect is a $\sqrt N$ decay. 

\begin{proposition}
	For any $N\geq S$, we have $C_{N,S}\geq 1/\sqrt{N}$ and $C_{N,S,\R}\geq 1/\sqrt{N}$. 
\end{proposition}

\begin{proof}
	Fix any non-zero $u\in\C^S$ and $\Omega$ of cardinality $S$. For any $\epsilon>0$, we define the measure $\mu_\epsilon=\sum_{j=1}^S u_j\delta_{\epsilon\omega_j}$. We see that $\hat{\mu_\epsilon}(k)\to \sum_{j=1}^S u_j$ for each $k\in\Z$ as $\epsilon\to 0$ and so
	\[
	\lim_{\epsilon\to 0} C_N(\mu_\epsilon) 
	=\lim_{\epsilon\to 0} \frac{|\hat{\mu_\epsilon}(0)|}{\|\hat\mu_\epsilon\|_{\ell^2_N}}
	=\frac{\big|\sum_{j=1}^S u_j\big|}{\sqrt{N} \big|\sum_{j=1}^S u_j\big|}
	=\frac{1}{\sqrt{N}}. 
	\]
	This argument also applies to when $u\in\R^S$.
\end{proof}

Obtaining an upper bound on $C_{N,S}$ that decays in $N$ (if this is even possible) appears to be a difficult problem and any solution should, in principle, address the number theoretic issues that might arise. Indeed, let $\mu\in\calM_S$ and assume that the amplitudes of $\mu$ are identically one. Then we have
\[
C_N(\mu)^2
=\frac{S^2}{\sum_{k=0}^{N-1}\Big|\sum_{j=1}^S e^{2\pi ik\omega_j} \Big|^2}
\]
The exponential sum $\big|\sum_{j=1}^S e^{2\pi ik\omega_j} \big|$ is $O(S)$ when its $S$ phases $\{e^{2\pi ik\omega_j}\}_{j=1}^S$ ``align" or occupy a small portion of the unit complex circle. In principle, there could exist very special $\Omega$ with particular number theoretic properties, such that the phases $\{e^{2\pi ik\omega_j}\}_{j=1}^S$ do not align for all $0\leq k\leq N-1$.

\section*{Acknowledgements} 

Weilin Li gratefully acknowledges support from the James C. Alexander prize and the AMS-Simons Travel Grant. Wenjing Liao is supported in part by the NSF grant DMS 1818751. All authors sincerely thank John J. Benedetto and C. Sinan G\"unt\"urk for their helpful suggestions and their encouragement. We also thank Dmitry Batenkov for his careful reading the first draft of this paper and for giving us insightful feedback.

\appendix

\section{Proof of lemmas and theorems}
\label{seclemmas}

\subsection{Proof of Lemma \ref{lemmapsi}}
\label{secapppsi}

Let $\Theta(U,\hat U) = \{\frac \pi 2 \ge \theta_1 \ge \theta_2 \ge \ldots \ge \theta_S\ge 0\}$ be the canonical angles between the subspaces spanned by the columns of $U$ and $\hat U$. Since ESPRIT is invariant to the choice of orthonormal basis, when we write $U$ and $\hat U$, we refer to the specific choice of bases for which their columns consist of the canonical vectors. In other words, we let $U = [u_1 \ u_2 \ \ldots \ u_S]$ and $\hat U = [\hat u_1 \ \hat u_2 \ \ldots \ \hat u_S]$, and assume
$\cos \theta_k = |u_k^* \hat u_k|, \quad\text{for } k =1 ,\ldots, S.$

We first derive several matrix perturbation bounds. The first one about $\Theta(U,\hat U)$ follows from the Wedin's theorem \cite{wedin1972perturbation,li1998relative}:

\begin{lemma}
	\label{lemmawedin}
	Fix positive integers $L,M,S$ such that \eqref{eq:L} holds. For any $\mu\in\calM_S$ and $\eta\in\C^{M+1}$ such that $2\|\calH(\eta)\|_2 \le  \sigma_S (\calH(y^0))$, we have
	$$\sin \theta_1 \le \frac{2\|\calH(\eta)\|_2}{\sigma_S (\calH(y^0))}.$$
\end{lemma}

Lemma \ref{lemmawedin} shows that when $U$ and $\hat U$ are chosen so that their columns align, the column spaces of $U$ and $\hat U$ are close when the noise is sufficiently small. This leads to the following perturbation bounds for $\|\hat U - U\|_2$ and $\|\hat \Psi -\Psi\|_2$, where the proofs can be found in Appendix \ref{seclemmas}.

\begin{lemma}
	\label{lemauper}
	Fix positive integers $L,M,S$ such that \eqref{eq:L} holds. For any $\mu\in\calM_S$ and $\eta\in\C^{M+1}$, if $2\|\calH(\eta)\|_2 \le  x_{\min} \sigma_{S}(\Phi_L)\sigma_{S}(\Phi_{M-L})$, then
	\[
	\|\hat U - U\|_2 
	\le \frac{2\sqrt{2S}\|\calH(\eta)\|_2}{x_{\min}\sigma_{S}(\Phi_L)\sigma_{S}(\Phi_{M-L})}.
	\]
\end{lemma}

\begin{proof}
	For $k=1,\ldots,S$,
	\begin{align*}
	\|\hat u_k - u_k\|_2^2 
	& =  4 \sin^2 \left( \frac{\theta_k}{2}\right) = 2(1-\cos \theta_k) \le 2(1-\cos^2\theta_k) \le 2\sin^2\theta_k.
	\end{align*}
	Using this inequality, see that 
	\[
	\|\hat U - U\|_2
	\leq \|\hat U- U\|_F
	= \(\sum_{k=1}^S \left \| \hat u_k - u_k\right\|^2_2\)^{1/2}
	\le \(2S \sin^2\theta_1\)^{1/2}
	= \sqrt{2S} \sin \theta_1.
	\]
	The proof is complete once we apply Lemma \ref{lemmawedin} and the inequality
	\[
	\|\calH(y^0)\|_2
	\geq x_{\min} \sigma_{S}(\Phi_L)\sigma_{S}(\Phi_{M-L}). 
	\]	
\end{proof}

Since $\Psi$ is computed from $U$, an estimate of $\|\widehat\Psi - \Psi\|_2 $ can be derived from $\|\hat U-U\|_2$ as follows.
\begin{lemma}
	\label{lemauper2}
	Fix positive integers $L,M,S$ such that \eqref{eq:L} holds. For any $\mu\in\calM_S$ and $\eta\in\C^{M+1}$ such that $\|\hat U-U\|_2\leq \sigma_{S}(U_0)/2$, one has
	\[
	\|\widehat\Psi - \Psi\|_2 
	\le  \frac{7\|\hat U-U\|_2}{\sigma_S^2(U_0)}.
	\]	
\end{lemma}

\begin{proof}
	By triangle inequalities, we obtain
	\begin{align*}
	\|\hat\Psi - \Psi \|_2 
	&=   \| (\hat U_0^{\dagger}  - U_0^{\dagger}) \hat U_1 + U_0^{\dagger}( \hat U_1 -   U_1) \|_2 \\
	&\le  \| \hat U_0^{\dagger}  - U_0^{\dagger}\|_2  \|\hat U_1\|_2 + \|U_0^{\dagger}\|_2 \| \hat U_1 -   U_1 \|_2 \\
	&\leq \| \hat U_0^{\dagger}  - U_0^{\dagger}\|_2  + \|U_0^{\dagger}\|_2 \| \hat U -   U \|_2,
	\end{align*}
	where for the last inequality, we used that $\hat U_1$ is the submatrix containing the last $L$ rows of $\hat U$ and that $\hat U$ has orthonormal columns, so $\|\hat U_1\|_2 \le \|\hat U \|_2 = 1$ and $\|\hat U_1-U_1\|_2\leq \| \hat U -   U \|_2$. Note that by assumption, we have
	\[
	\|\hat U_0-U_0\|_2
	\leq \|\hat U-U\|_2
	\leq \frac{1}{2\sigma_{S}(U_0)}. 
	\]
	This enables us to apply \cite[Theorem 3.2]{hansen1987truncatedsvd}, and so
	\[
	\| \hat U_0^{\dagger}  - U_0^{\dagger}\|_2  \le  \frac{3\|\hat U_0 - U_0\|_2}{\sigma_S(U_0) \big(\sigma_S(U_0) - \|\hat U_0 - U_0\|_2\big)}
	\leq \frac{6\|\hat U - U\|_2}{\sigma_{S}^2(U_0)}.
	\]
	Therefore, we have that
	\begin{align*}
	\|\hat\Psi - \Psi \|_2  
	\le 
	\(\frac{6}{\sigma_S^2(U_0)}+\frac{1}{\sigma_{S}(U_0)}\) \|\hat U - U\|_2 
	\leq \frac{7\|\hat U - U\|_2}{\sigma_{S}^2(U_0)}.
	\end{align*}

\end{proof}

	Combining Lemma \ref{lemauper} and Lemma \ref{lemauper2} yields Lemma \ref{lemmapsi}.

\subsection{Proof of Equation \eqref{lemmamatching3}}
\label{appeq}

\begin{proof}
	If $\md(\Psi,\hat\Psi)\geq 1$, we have the trivial inequality 
	\[
	\md(\Omega,\hat\Omega)\leq \frac{1}{2}
	\leq \frac{1}{2}\md(\Psi,\hat\Psi).
	\]
	If ${\rm md}(\Psi,\hat\Psi)\leq 1$, we have 
	\begin{align}
	\label{eq:sine}
	2 \pi\,  {\rm md}(\Omega,\hat\Omega)
	\leq 2\pi \, \sin^{-1}\big({\rm md}(\Psi,\hat\Psi) \big)\leq \frac{2}{\pi}\, {\rm md}(\Psi,\hat\Psi). 
	\end{align}
	The first inequality is a consequence of the law of sines, see Figure \ref{fig:lawsine}. The second inequality follows by observing that the function $f(t)=\sin(\pi t/2)$ is concave on the domain $0\leq t\leq 1$ and $f(0)=0$ and $f(1)=1$, so $f^{-1}(u)=2\sin^{-1}(u)/\pi\leq u$ for all $0\leq u\leq 1$. Thus, we have proved Equation \eqref{lemmamatching3}.
	
	\begin{SCfigure}
			\centering
		\begin{tikzpicture}[scale=0.75]
		\draw[thick] (5.1,0)  arc[radius = 5cm, start angle= 0, end angle= 60];
		\draw[thick] (0,0) -- (7,5) -- (5,1) -- (0,0);
		\filldraw 
		(0,0) circle (1pt) node[align=right ] {$0$ \hspace{1em} } --
		(5,1) circle (1pt) node[align=right ] {\hspace{3em} $e^{2\pi i\omega_j}$} -- 
		(7,5) circle (1pt) node[align=right]{\hspace{2.5em} $\hat\lambda_{\psi(j)}$};
		\filldraw 
		(4.13,2.95) circle (1pt) node[align=right]{$e^{2\pi i\hat \omega_{\psi(j)}}$ \hspace{4em}};
		\draw (1.4,0.6) node {$\theta_j$};
		\draw (6.2,4) node {$\alpha_j$};
		\draw (3,0.3) node {$1$};
		\draw (6.3,2.8) node {$\epsilon_j$};
		\end{tikzpicture}
		\caption{Geometric figure corresponding to inequality \eqref{eq:sine}. Here, we define $\theta_j=2\pi |\omega_j-\hat \omega_{\psi(j)}|_\T$ and $\epsilon_j=|e^{2\pi i\omega_j}-\hat\lambda_{\psi(j)}|$. Note that $\epsilon_j\leq {\rm md}(\Psi,\hat\Psi)$. By the law of sines, we have $\epsilon_j/\sin(\theta_j)=1/\sin(\alpha_j)$. This implies $\sin(\theta_j)= \epsilon_j\sin(\alpha_j)\leq \epsilon_j$.}
		\label{fig:lawsine}
	\end{SCfigure}
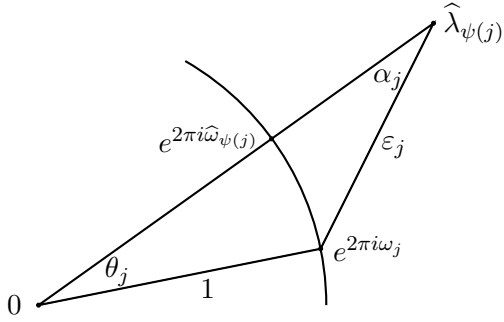
\end{proof}

\subsection{Proof of Lemma \ref{lemmamatching}}
\label{appmd}

\subsubsection{Part (a)}
\begin{proof}
	Recall from Equation \eqref{eq:P2} that $\Psi$ is diagonalizable by an invertible matrix $P^{-1}$. Thanks to \eqref{lemmamatching3} and the Bauer-Fike theorem, see \cite{bauer1960norms} and \cite[Chapter IV, Theorem 3.3]{stewart1990matrix}, we have
	\[
	{\rm md}(\Omega,\hat\Omega) 
	\leq \frac{1}{2}{\rm md}(\Psi,\hat\Psi) 
	\le \frac{1}{2} (2S-1) \kappa(P^{-1}) \|\hat\Psi -\Psi\|_2.
	\]
	To control the conditioning of $P$, we use \eqref{eq:P} to see that $P=\Phi_L^\dagger U$ and so
	\begin{equation}
	\label{eq:condP}
	\kappa(P^{-1})
	=\|P\|_2\|P^{-1}\|_2
	\leq \|\Phi_L^\dagger\|_2\|\Phi_L\|_2
	=\kappa(\Phi_L)
	\leq \frac{\|\Phi_L\|_F}{\sigma_{S}(\Phi_L)}
	\leq \frac{\sqrt{(L+1)S}}{\sigma_{S}(\Phi_L)}. 
	\end{equation}
	Combining the two inequalities together completes the proof of Lemma \ref{lemmamatching} Part (a). 
\end{proof}

\subsubsection{Part (b)}

	The argument is identical to the one in \cite[Chaper IV, Theorem 2.3]{stewart1990matrix} holds for general matrices. For the reader's convenience, we reproduce the argument in the context of our setting.
	
\begin{proof}
	Recall from \eqref{eq:P2} that there is an invertible matrix $P$ such that $P\Psi P^{-1}=D_{\Omega}$. For convenience, let $E:=\hat\Psi-\Psi$ so that $\hat\Psi =\Psi+E$. We consider the matrix 
	\[
	J
	:=P\hat\Psi P^{-1}
	=D_\Omega+Q,
	\quad\text{where}\quad 
	Q:=PEP^{-1}. 
	\]
	Note that $J$ has the same eigenvalues as that of $\hat\Psi$. Define the quantity 
	\[
	\epsilon:=
	\frac{\sqrt{(L+1)S}\|E\|_2}{\sigma_{S}(\Phi_L)}.
	\]
	Note that we have 
	\[
	{\epsilon
	=\frac{\|\Phi_L\|_F \|E\|_2}{\sigma_{S}(\Phi_L)}
	}
	\geq \kappa(\Phi_L)\|E\|_2
	\geq \kappa(P)\|E\|_2
	\geq \|P\|_2\|E\|_2\|P^{-1}\|_2,
	\] 
	where we used the inequality $\kappa(P)\leq \kappa(\Phi_L)$ which was established in \eqref{eq:condP}. Consequently, we have $|Q_{j,j}|\leq \|E\|_2\leq \epsilon$ and $|Q_{j,k}|\leq \epsilon$. 	
	
	From here onwards, fix an index $1\leq k\leq S$. We introduce a parameter $\alpha>0$ that will be chosen later. Let $A$ be the diagonal matrix whose $k$-th diagonal is $\alpha$ and the remaining diagonal entries are $1$. If we define $J_\alpha:=AJA^{-1}$, then the eigenvalues of $\hat\Psi$, $J$, and $J_\alpha$ coincide. A direct calculation shows that this conjugation by $A$ multiplies the $k$-th column of $J$ by $\alpha^{-1}$, and multiplies the $k$-th row of $J$ by $\alpha$. Explicitly, we have
	\[
	J_\alpha
	:= AJA^{-1}
	=D_\Omega +
	\begin{pmatrix}
	Q_{1,1} & & &\alpha^{-1}Q_{1,k} \\
	&\ddots  & &\vdots \\
	& &Q_{k-1,k-1} &\alpha^{-1} Q_{k-1,k} \\
	\alpha Q_{k,1} &\cdots &\alpha Q_{k,k-1}  &Q_{k,k} &\alpha  Q_{k,k+1} &\cdots &\alpha Q_{k,S}\\
	& & &\alpha^{-1} Q_{k,k+1} &Q_{k+1,k+1} \\
	& & &\vdots  & &\ddots \\
	& & &\alpha^{-1} Q_{S,k} & & &Q_{S,S}
	\end{pmatrix}
	\]
	This equation tells us that the $k$-th Gershgorin disk of $J_\alpha$ is centered at $e^{-2\pi i\omega_k}+Q_{k,k}$ with radius bouned above by $(S-1)\alpha \epsilon$, while the $j$-th Gershgorin disk for $j\not =k$ is centered at $e^{-2\pi i\omega_j}+Q_{j,j}$ with radius bounded by $\alpha^{-1}\epsilon+(S-2)\epsilon$. This implies that, the $k$-th disk is contained in the disk centered at $e^{-2\pi i\omega_k}$ and radius $\epsilon+(S-1)\alpha\epsilon$, while the remaining disks are contained in disks centered at $e^{-2\pi i\omega_j}$ with radius $\alpha^{-1}\epsilon+(S-2)\epsilon+\epsilon$ for $j \neq k$. Thus, a sufficient condition for the $k$-th Gershgorin disk to be disjoint from the other disks is that 
	\begin{equation}
	\label{eq:diskcondition}
	\text{for all } j\not=k, \quad |e^{-2 \pi i\omega_k}-e^{-2 \pi i\omega_j}|
	\geq (S-1)\alpha \epsilon +\alpha^{-1}\epsilon+S\epsilon. 
	\end{equation}
	In particular, using the inequality $|1-e^{2\pi it}|\geq 4|t|$ for all $|t|\leq 1/2$, we see that \eqref{eq:diskcondition} holds provided that 
	\begin{equation}
	\label{eq:diskcondition2}
	4\Delta \geq (S-1)\alpha \epsilon+\alpha^{-1}\epsilon+S\epsilon,
	\end{equation}
	where we recall that $\Delta:=\min_{j \neq k}|\omega_j-\omega_k|_\T$. From now onwards, we pick $\alpha =\epsilon/\Delta$ so that the above condition is equivalent to
	\[
	\frac{(S-1)\epsilon^2}{\Delta^2}+\frac{S\epsilon}{\Delta} \leq 3. 
	\]
	We see that this inequality is satisfied because assumption \eqref{eq:smallnoise} implies $\epsilon/\Delta\leq 1/S$. 
	
	There is an eigenvalue $\hat\lambda_k$ of $\hat\Psi$ contained in the $k$-th Gershgorin disk, which has center $e^{2\pi i\omega_k}+Q_{k,k}$ and radius bounded by $(S-1)\epsilon^2/\Delta$. Since $|Q_{k,k}|\leq \|E\|_2$, we have
	\begin{equation}
	\label{eq:diskcondition3}
	|e^{2\pi i\omega_k}-\hat\lambda_k|
	\leq \|\Psi-\hat\Psi\|_2+\frac{(L+1)S^2}{\sigma_{S}^2(\Phi_L)\Delta}\|\Psi-\hat\Psi\|^2_2. 
	\end{equation}
	This shows that for the $k$-th eigenvalue of $\Psi$, there is a $\hat \lambda_k$ that satisfies this inequality. 
	
	Furthermore, we check that the assumption \eqref{eq:smallnoise} together with \eqref{eq:diskcondition3} implies that
	\[
	|e^{2\pi i\omega_k}-\hat\lambda_k|
	\leq 2\|\Psi-\hat\Psi\|_2. 
	\]
	In particular, the assumption \eqref{eq:smallnoise} yields
	\begin{align*}
	|e^{2\pi i\omega_k}-\hat\lambda_k|
	\leq 2\|\Psi-\hat\Psi\|_2
	\leq \frac{2\sigma_{S}^2(\Phi_L)\Delta}{S^2(L+1)} 
	=\frac{2\Delta}{S}\frac{\sigma_{S}^2(\Phi_L)}{\|\Phi\|_F^2}
	\leq \frac{1}{2S} \min_{j\not=k} |e^{-2\pi i\omega_k}-e^{-2\pi i\omega_j}|. 
	\end{align*}
	This shows that the best matching between the eigenvalues of $\Psi$ and $\hat\Psi$ satisfy
	\[
	{\rm md}(\Psi,\hat\Psi)
	\leq 2\|\Psi-\hat\Psi\|_2. 
	\]
	Using \eqref{lemmamatching3} completes the proof.
\end{proof}

\subsection{Proof of Lemma \ref{lemmauncertainty}}
\label{appu0}

\begin{proof}
	We first write
	$$
	U 
	= \begin{bmatrix} w_1^* \\ U_1 \end{bmatrix}
	= \begin{bmatrix} U_0 \\ w_0^* \end{bmatrix}.
	$$
	Since $U^* U = I_S$, we have $U_j^* U_j = I_S - w_j w_j^*$ for $j=0,1$. 
	This implies $w_j$ is an eigenvector of $U_j^* U_j$ associated with eigenvalue $1-\|w_j\|_2^2$, and the other $S-1$ eigenvalues of $U_j^* U_j$ are $1$ with eigenvectors perpendicular to $w_j$. This shows that for $j=0,1$, we have 
	\begin{equation}
	\label{eq:sigS}
	\sigma_1 (U_j)= \sigma_2 (U_j)= \cdots =\sigma_{S-1} (U_j)= 1 \quad\text{and}\quad 
	\sigma_{S} (U_j)= \sqrt{1-\|w_j\|_2^2}. 
	\end{equation}
	It suffices to provide an upper bound for $\|w_j\|_2$ that is better than the trivial one $\|w_j\|_2\leq 1$. This  can be loosely interpreted as the amount of information contained in the first or last row of $U$. 
	
	The column space of $\Phi_L$ is a $S$-dimensional complex subspace of $\C^{L+1}$, which we denote by $Z$. Let $e_j$ be the $j$-th canonical basis vector in $\C^{L+1}$, for $0\leq j\leq L$. By definition, the columns of $U$ is an orthonormal basis for the range of $\Phi_L$. The norm of $w_0$ (respectively $w_1$) is the length of projection of $e_{L}$ (respectively $e_0$) onto $Z$. Hence,
	\[
	\|w_1\|_2
	=\max_{z\in Z,\ z\not=0} \frac{|e_0^* z|}{\|z\|_2}, 
	\quad\text{and}\quad  
	\|w_0\|_2
	=\max_{z\in Z,\ z\not=0} \frac{|e_L^* z|}{\|z\|_2}.  
	\]
	For each $z\in Z$, there are coefficients $c_j=c_j(z)$ such that $z =\sum_{j=1}^S c_j \phi_L(\omega_j)$ where $\phi_L(\omega_j)$ is the $j$-th column of $\Phi_L$. Another way of looking at this is to define the measure $\mu=\sum_{j=1}^S c_j\delta_{\omega_j}$, and we readily see that $z_k=\hat\mu(k)$ for $k=0,1,\dots,L$. Using this observation, we see that 
	\begin{equation}
	\label{eq:w1}
	\|w_1\|^2_2
	=\sup_{\supp(\mu)\subset\Omega} \frac{|\hat\mu(0)|^2}{\sum_{k=0}^{L} |\hat\mu(k)|^2}.
	\end{equation}
	We also have that,
	\begin{equation}
	\label{eq:w0} 
	\|w_0\|^2_2
	=\sup_{\supp(\mu)\subset\Omega} \frac{|\hat\mu(L)|^2}{\sum_{k=0}^{L} |\hat\mu(k)|^2}
	=\sup_{\supp(\mu)\subset-\Omega}\frac{|\hat\mu(0)|^2}{\sum_{k=0}^{L} |\hat\mu(k)|^2},
	\end{equation}
	where second equality is a consequence of the following observation: if $\mu=\sum_{j=1}^S u_j\delta_{\omega_j}$, then we define $\nu=\sum_{j=1}^S u_j\delta_{-\omega_j}e^{2\pi iL\omega_j}$ and check that $\hat\mu(k)=\hat\nu(L-k)$ for all $k\in\Z$.
	
	The above identities relate $\sigma_S(U_0)$ and $\sigma_S(U_1)$ to the quantities defined on the right hand side of equations \eqref{eq:w0} and \eqref{eq:w1}, which are interpreted as the concentration of $\hat\mu$ in its zero-th Fourier coefficient relative to the total energy contained in its first $L+1$ Fourier coefficients. We estimate these concentrations using two different methods. 
	
	One approach is to control these in terms of $\Phi_L(\Omega)$. By application of Cauchy-Schwarz, 
	\[
	\sup_{\supp(\mu)\subset \Omega}\frac{|\hat\mu(0)|^2}{\sum_{k=0}^{L} |\hat\mu(k)|^2}
	\leq \frac{S}{\sigma_S^2(\Phi_L(\Omega))}.
	\]
	We obtain a similar inequality  for when $\mu$ is supported in $-\Omega$ because $\sigma_{S}(\Phi_L(\Omega))=\sigma_{S}(\Phi_L(-\Omega))$. Using this with the above inequalities \eqref{eq:sigS}, \eqref{eq:w1} and \eqref{eq:w0} yields, 
	\[
	\min(\sigma_{S}^2(U_0),\sigma_{S}^2(U_1))
	\geq 1-\frac{S}{\sigma_S^2(\Phi_L(\Omega))}.
	\]
	
	Another approach is to interpret the right hand sides of \eqref{eq:w1} and \eqref{eq:w0} in terms of an uncertainty principle. It immediately follows from Theorem \ref{thm:UPcomplex} and the previous inequalities \eqref{eq:sigS}, \eqref{eq:w1} and \eqref{eq:w0} that
	\[
	\min(\sigma_{S}^2(U_0),\sigma_{S}^2(U_1))
	\geq 4^{-S}.
	\]
\end{proof}

\subsection{Proof of Theorem \ref{thm1}}
\label{secproofthm1}

\begin{proof}
	\begin{enumerate}[(a)]
		\item 
		Note that the noise assumption \eqref{thm1con} guarantees that $\|U -\hat U\|_2 \le \sigma_S(U_0)/2$ so the assumption in Lemma \ref{lemauper2} holds. Thus, we combine Lemma \ref{lemmamatching} Part (a), Lemma \ref{lemauper}, and Lemma \ref{lemauper2} to conclude that
		\begin{align*}
		\md(\Omega,\hat\Omega)
		&\leq \frac{20 S^2\sqrt{L+1}}{x_{\min}\sigma_{S}^2(\Phi_L)\sigma_{S}(\Phi_{M-L})\sigma_{S}^2(U_0)}\|\calH(\eta)\|_2. 
		\end{align*}
		\item 
		We readily check that the noise assumption \eqref{thm1con2} implies the conditions of Lemma \ref{lemmamatching} Part (b), Lemma \ref{lemauper}, and \ref{lemauper2}, are satisfied. Hence, we combine the aforementioned lemmas with Lemma \ref{lemmamatching} to complete the proof.
	\end{enumerate}	
	
\end{proof}

\subsection{Proof of Theorem \ref{thmesprit}}
\label{secproofthmesprit}

\begin{proof}
According to Theorem \ref{thmsin}, if $\Omega$ satisfies Assumption \ref{def:clumps} with parameters $(M/2,S,\alpha,\beta)$  for some $\alpha>0$ and $\beta$ satisfying \eqref{eq:sep2}, then there exist explicit constants $c_a:=C_a(\lambda_a,M/2)$ such that 
$$
\hspace{-.5em} \sigma_{\min}(\Phi_{M/2})
\geq \sqrt{\frac M 2}\(\sum_{a=1}^A \big( c_a \alpha^{-\lambda_a+1} \big)^2 \)^{-\frac{1}{2}}. 
$$
Combining the inequality above, Lemma \ref{lemmauncertainty} and Theorem \ref{thm1} gives rise to Theorem \ref{thmesprit}.

\end{proof}

\subsection{Proof of Theorem \ref{thm:sep}}

\begin{proof}
	The proof amounts to checking that the assumptions of the first statement in Theorem \ref{thmesprit} hold, and then using known results for $\sigma_{S}(\Phi_L)$. First, the assumptions on $M,L,S$ imply that $S\leq L\leq M-L+1$. Second, the assumption $\Delta\geq C/L$ together with inequality \eqref{eq:moitra} yields
	\begin{equation}
	\label{eq:phiL}
	\sigma_{S}^2(\Phi_{M-L})
	=\sigma_{S}^2(\Phi_{L})
	\geq \frac{C-1}{C} \, L.
	\end{equation}
	Combining Lemma \ref{lemmauncertainty} and \eqref{eq:phiL}, we have that 
	\begin{equation}
	\label{eq:U}
	\min\big(\sigma_{S}^2(U_0),\sigma_{S}^2(U_1)\big)
	\geq 1-\frac{S}{\sigma_{S}^2(\Phi_L)}
	\geq 1-\frac{C}{C-1}\, \frac{S}{L}. 
	\end{equation}
	Note that the term on the right hand side is strictly positive due to the assumptions that $M=2L\geq 4S$ and $C>2$. The noise assumption \eqref{eq:noise} together with inequalities \eqref{eq:phiL} and \eqref{eq:U} imply that the assumptions of the first statement of Theorem \ref{thm1} are satisfied, which gives us the inequality,
	\[
	\text{md}(\Omega,\hat\Omega)
	\leq \frac{20\, S^2 \sqrt{L+1}\,\|\calH(\eta)\|_2}{x_{\min}\sigma_S^2(U_0) \sigma_S^3(\Phi_L)}.
	\]
	Inserting \eqref{eq:phiL}, and \eqref{eq:U} into this inequality completes the proof of the theorem. 
\end{proof}

\bibliographystyle{plain}
\bibliography{SRlimitFourierbib}

\end{document}